\documentclass[11pt]{article}

\usepackage[T1]{fontenc}
\usepackage{lmodern}
\usepackage[margin=1in]{geometry}
\usepackage{microtype}
\usepackage{mathtools,amssymb,amsthm}
\usepackage{booktabs}
\usepackage{array}
\usepackage{enumitem}
\usepackage{xcolor}
\usepackage[
  colorlinks=true,
  linkcolor=blue!45!black,
  citecolor=blue!45!black,
  urlcolor=blue!45!black,
  pdftitle={Lower Bounds for Domination-Type Problems Parameterized by Rank-Width},
  pdfauthor={Chenghua Liu and Boning Meng}
]{hyperref}
\usepackage{graphicx}
\usepackage{tikz}
\usetikzlibrary{shapes.geometric,positioning}
\usepackage{authblk}
\setlist{nosep,leftmargin=*}
\allowdisplaybreaks
\newcolumntype{L}[1]{>{\raggedright\arraybackslash}p{#1}}

\newtheorem{theorem}{Theorem}[section]
\newtheorem{lemma}[theorem]{Lemma}
\newtheorem{corollary}[theorem]{Corollary}
\newtheorem{proposition}[theorem]{Proposition}
\theoremstyle{remark}
\newtheorem{remark}[theorem]{Remark}
\numberwithin{equation}{section}

\newcommand{\F}{\mathbb F_2}
\newcommand{\rk}{\operatorname{rank}_{\F}}
\newcommand{\rw}{\operatorname{rw}}
\newcommand{\lrw}{\operatorname{lrw}}
\newcommand{\width}{\operatorname{width}}
\newcommand{\Zle}{Z_{\leq}}
\newcommand{\Zeq}{Z_{=}}
\newcommand{\IDS}{\textnormal{Independent Dominating Set}}
\newcommand{\CDS}{\textnormal{Connected Dominating Set}}
\newcommand{\TDS}{\textnormal{Total Dominating Set}}

\title{Lower Bounds for Domination-Type Problems\\
Parameterized by Rank-Width}

\author[1]{Chenghua Liu}
\author[2]{Boning Meng}

\affil[1]{Institute of Software, Chinese Academy of Sciences, Beijing, China}
\affil[2]{University of Regensburg, Regensburg, Germany}
\affil[ ]{\texttt{liuch.russell@gmail.com}, \texttt{mengboning2013@gmail.com}}

\date{}

\begin{document}
\maketitle

\begin{abstract}
For graphs of rank-width \(w\), the algorithms of Bui-Xuan, Telle, and Vatshelle (\emph{Theor. Comput. Sci.}, 2013) for fixed finite/cofinite \((\sigma,\rho)\)-problems and of Bergougnoux and Kant\'e (\emph{SIAM J. Discrete Math.}, 2021) for Connected Dominating Set run in \(2^{O(w^2)}n^{O(1)}\) time. Bergougnoux, Korhonen, and Nederlof (STACS 2023) proved a matching lower bound under the Exponential Time Hypothesis (ETH) for \emph{Weighted} Dominating Set, but left the unweighted problem open. We prove that, unless ETH fails, Dominating Set admits no \(2^{o(w^2)}n^{O(1)}\)-time algorithm, even on split graphs and, separately, on bipartite graphs of diameter at most four, and even with a rank-decomposition or witnessing vertex order supplied. The proof replaces the earlier weights by a two-guard gadget and uses a low-rank equality gadget to carry \(k^2\) assignment bits through cuts of rank \(O(k)\). The construction also gives the same lower bound for Independent, Connected, and Total Dominating Set on restricted graph classes and applies to a broad family of \((\sigma,\rho)\)-set problems. This family includes cases in which \(\sigma\) is neither finite nor cofinite and contains the entire nontrivial cofinite--cofinite minimization regime. Every solution within the target budget has target size and corresponds bijectively to a satisfying assignment. Under the counting Exponential Time Hypothesis (\(\#\mathrm{ETH}\)), the same bounds therefore hold for counting solutions of size at most or exactly the target. Together with the known algorithms, our results show that the quadratic dependence on the rank-width \(w\) is optimal up to constant factors in the exponent for the classical problems above and throughout the covered finite/cofinite regime.
\end{abstract}
\section*{Acknowledgments}
The authors used OpenAI's ChatGPT in preparing
this manuscript, including for language editing and \LaTeX{} preparation.
ChatGPT also contributed to the exploratory development of most of the
arguments. All claims and proofs were independently checked and finalized
by the authors, who take full responsibility for the content.
\section{Introduction}
\label{sec:introduction}

Domination-type problems provide a natural meeting point of structural graph theory, exact and parameterized algorithms, and counting complexity. Two kinds of structure matter simultaneously.  The local domination rule determines which neighbor counts are permitted, while the host graph class determines which global adjacency patterns are available.  Width parameters connect these perspectives by measuring the information that must cross a decomposition cut.

This paper studies how much information domination algorithms must retain across cuts of small rank over \(\mathbb F_2\).  Its main conclusion is that a cut of rank \(w\) can encode genuinely quadratic information: the familiar \(2^{O(w^2)}\) algorithms cannot in general be improved to \(2^{o(w^2)}n^{O(1)}\), even for unweighted problems, for restricted graph classes, and for the corresponding problems of counting solutions up to, or exactly at, the target size.  We first review domination theory and the \((\sigma,\rho)\)-framework, then place the result within the algorithmic theory of rank-width.

\subsection{Domination and the \texorpdfstring{\((\sigma,\rho)\)}{(sigma,rho)}
framework}

A set \(D\subseteq V(G)\) is a \emph{dominating set} of a graph \(G\) if every vertex outside \(D\) has a neighbor in \(D\).  Equivalently, \(N_G[v]\cap D\neq\varnothing\) for every \(v\in V(G)\).  The minimum size of such a set is the domination number \(\gamma(G)\).  Domination is one of the central vertex-subset themes in graph theory: it models the placement of facilities so that every location is selected or reaches a selected location.  The subject has developed into a large theory of structural characterizations, algorithms, complexity, and variants.  The companion volumes \emph{Fundamentals of Domination in Graphs} and \emph{Domination in Graphs: Advanced Topics} provide standard accounts of this theory and its terminology \cite{haynes-hedetniemi-slater-1998,haynes-hedetniemi-slater-advanced-1998}.

The basic requirement of a dominating set can be strengthened in several ways.  A solution may be required to be independent or connected, as in Independent Dominating Set and Connected Dominating Set; selected vertices may themselves have to be dominated, as in Total Dominating Set; every vertex may require \(q\) selected representatives, as in \(q\)-Tuple Dominating Set for a fixed integer \(q\geq1\); or an unselected vertex may be required to have exactly one selected neighbor, as in Perfect Dominating Set. These choices lead to different structural parameters and substantially different algorithmic behavior. Independent and total domination are already classical examples \cite{cockayne-hedetniemi-1976,cockayne-dawes-hedetniemi-1980}, while the two companion volumes above survey many further forms \cite{haynes-hedetniemi-slater-1998,haynes-hedetniemi-slater-advanced-1998}.

The associated optimization problem is computationally difficult.  In particular, the decision version of Dominating Set remains NP-complete on split graphs and on bipartite graphs~\cite{bertossi-1984}.  For associated counting problems, the host graph class is especially important, and graph classes that are all highly structured can lie on opposite sides of the counting-complexity boundary.  Kijima, Okamoto, and Uno gave polynomial-time algorithms for counting dominating sets and minimum dominating sets on interval and trapezoid graphs.  On split graphs and chordal bipartite graphs, counting all dominating sets is \(\#\mathrm{P}\)-complete, whereas counting minimum dominating sets is \(\#\mathrm{P}\)-hard~\cite{kijima-okamoto-uno-2011}.  Further exact counting hardness is known under strong restrictions.  Hunt et al.\ proved \(\#\mathrm{P}\)-completeness for counting minimum dominating sets in planar graphs and for counting dominating sets in planar bipartite graphs~\cite{hunt-et-al-1998}.  In a recent preprint, Zheng and Meng proved, under polynomial-time Turing reductions, that counting all dominating sets and counting all total dominating sets are \(\#\mathrm{P}\)-complete even on 3-regular planar bipartite simple graphs~\cite{zheng-meng-2026}.

Many variants differ only in the number of selected neighbors allowed at a vertex. Let \(\mathbb N_{\geq q}:=\{q,q+1,\ldots\}\). Telle's state-based framework, now commonly written as the \((\sigma,\rho)\)-framework, captures these local constraints uniformly~\cite{telle-1994}.  For the nonempty sets considered here, let \(\sigma,\rho\subseteq\mathbb N_{\ge 0}\). A set \(S\subseteq V(G)\) is a \emph{\((\sigma,\rho)\)-set} if
\begin{equation}
  |N_G(v)\cap S|\in
  \begin{cases}
    \sigma,&v\in S,\\
    \rho,&v\notin S.
  \end{cases}
  \label{eq:sigma-rho-definition}
\end{equation}
For example, Dominating Set is the pair \((\mathbb N_{\ge 0},\mathbb N_{\ge 1})\), Independent Dominating Set is \((\{0\},\mathbb N_{\ge 1})\), Total Dominating Set is \((\mathbb N_{\ge 1},\mathbb N_{\ge 1})\), and \(q\)-Tuple Dominating Set for fixed \(q\geq1\) is \((\mathbb N_{\ge q-1},\mathbb N_{\ge q})\).  Perfect Dominating Set and Perfect Code correspond to \((\mathbb N_{\ge 0},\{1\})\) and \((\{0\},\{1\})\), respectively.  The framework also contains induced bounded- or prescribed-degree problems; Bui-Xuan, Telle, and Vatshelle developed decomposition algorithms for its locally checkable vertex-subset and vertex-partitioning generalizations~\cite{bui-xuan-telle-vatshelle-2013}.

The counting complexity of this framework has recently received systematic attention.  Assuming the counting Strong Exponential Time Hypothesis (\(\#\mathrm{SETH}\)), Focke et al.\ determined the optimal exponential base for counting \((\sigma,\rho)\)-sets of prescribed size on graphs supplied with a tree decomposition, for every fixed nontrivial finite/cofinite pair~\cite{focke-et-al-2023}.  Zheng and Meng introduced the broader Counting General Dominating Set (\(\#\)GDS) framework, a weighted counting framework that encompasses \(\#(\sigma,\rho)\)-Set~\cite{zheng-meng-2026}.  For maximum induced matching width (mim-width), for every nonempty \(\sigma,\rho\subseteq\mathbb N_{\ge0}\) with \(0\notin\rho\), the problem of finding a minimum-cardinality \((\sigma,\rho)\)-set is \(\mathrm{W}[1]\)-hard when parameterized jointly by the width \(w\) of a supplied linear branch decomposition and the target solution size; under ETH it admits no \(f(w)n^{o(w/\log w)}\)-time algorithm for any computable \(f\).  For finite/cofinite pairs, this is within a logarithmic factor in the exponent of the known \(n^{O(w)}\) XP upper bound~\cite{bakkane-jaffke-2022}. Our focus is the qualitatively different behavior of rank-width.

\subsection{Rank-width and domination algorithms}

Width parameters expose a graph through a sequence of controlled separations.  Treewidth is especially effective on sparse graphs, whereas clique-width and rank-width also remain bounded on many dense graph classes.  Rank-width was introduced by Oum and Seymour through the rank over \(\mathbb F_2\) of the adjacency matrix across a cut~\cite{oum-seymour-2006}.  It has the same bounded graph classes as clique-width and satisfies
\begin{equation}
  \rw(G)\leq \operatorname{cw}(G)
  \leq 2^{\rw(G)+1}-1.
  \label{eq:rw-cw-relation}
\end{equation}
This relation is due to Oum and Seymour~\cite{oum-seymour-2006}.
Rank-width has particularly strong decomposition algorithms.  For every fixed \(k\), one can in \(O(|V(G)|^3)\) time either construct a rank-decomposition of width at most \(k\) or certify that \(\rw(G)>k\)~\cite{hlineny-oum-2008}.  Oum also gave an \(O(8^k|V(G)|^4)\)-time certifying approximation that either outputs a rank-decomposition of width at most \(3k+1\) or certifies that \(\rw(G)>k\)~\cite{oum-2008,oum-rankwidth-survey-2017}.  The path-like analogue of rank-width is linear rank-width, denoted by \(\lrw\), and satisfies \(\rw(G)\leq\lrw(G)\).

General monadic second-order logic over vertex sets (MSO\(_1\)) gives broad tractability on graph classes of bounded clique-width when a suitable expression can be constructed~\cite{courcelle-makowsky-rotics-2000}.  However, translating a width-\(k\) rank-decomposition into a clique-width \((2^{k+1}-1)\)-expression does not by itself yield the \(2^{O(k^2)}\) dependence obtained by direct neighborhood-union dynamic programming~\cite{oum-seymour-2006,oum-saether-vatshelle-2014}.  Bui-Xuan, Telle, and Vatshelle introduced Boolean-width~\cite{bui-xuan-telle-vatshelle-2011} and developed neighborhood-union dynamic programming on graph decompositions~\cite{bui-xuan-telle-vatshelle-2010,bui-xuan-telle-vatshelle-2013}.  For every fixed pair for which each of \(\sigma\) and \(\rho\) is finite or cofinite, their framework gives a \(2^{O(k^2)}|V(G)|^{O(1)}\)-time algorithm from a supplied width-\(k\) rank-decomposition~\cite{bui-xuan-telle-vatshelle-2013}.  Together with Oum's explicit approximation above, this yields the rank-width bound
\begin{equation}
  2^{O(\rw(G)^2)}|V(G)|^{O(1)}.
  \label{eq:general-rw-upper-bound}
\end{equation}
For connectivity variants, Bergougnoux and Kant\'e used \(d\)-neighbor equivalence and representative sets to obtain the same \(2^{O(\rw(G)^2)}\) dependence for Connected Dominating Set and a wider class of connected or acyclic vertex-subset problems~\cite{bergougnoux-kante-2021}.  Their subsequent erratum affects only the stated polynomial dependence for certain acyclic \(q\)-partition applications, not the Connected Dominating Set result~\cite{bergougnoux-kante-erratum-2024}.

The quadratic exponent is known to be necessary for several problems. Bergougnoux, Korhonen, and Nederlof proved that, under ETH, Independent Set, Weighted Dominating Set, Maximum Induced Matching, and Feedback Vertex Set admit no \(2^{o(\lrw(G)^2)}n^{O(1)}\)-time algorithms~\cite{bergougnoux-korhonen-nederlof-2023}.  Their weighted domination reduction uses prohibitive weights to enforce its normal form.  They left as an explicit open problem whether the upper bound \eqref{eq:general-rw-upper-bound} is optimal for the ordinary \emph{unweighted} Dominating Set problem.  This is the starting point of the present work.

\subsection{Our results}
\label{subsec:introduction-results}

Our results are most compactly stated for a supplied vertex order.  If \(\pi\) is an order of \(V(G)\), let \(\width_G(\pi)\) be the maximum binary cut-rank of a prefix of \(\pi\).  Thus \(\rw(G)\leq\lrw(G)\leq\width_G(\pi)\).

\begin{theorem}[Main theorem, informal]
\label{thm:main-informal}
Assuming ETH, none of the following admits an algorithm with running time
\[
  2^{o(w^2)}|V(G)|^{O(1)},
\]
where \(w\) is the width of a supplied vertex order:
\begin{enumerate}[label=\textup{(\roman*)}]
  \item unweighted Dominating Set on monopolar graphs, on split graphs, or on bipartite graphs of diameter at most four;
  \item Independent Dominating Set on monopolar graphs;
  \item Connected Dominating Set and Total Dominating Set on split graphs, and separately on bipartite graphs of diameter at most four; and
  \item \((\sigma,\rho)\)-Set for every fixed decidable pair \((\sigma,\rho)\) such that \(\rho\) is cofinite and \(0\notin\rho\), in either of the following cases:
  \begin{enumerate}
  	\item \(0\in\sigma\) and \(1\in\rho\) (the result holds on monopolar graphs);
  	\item \(\sigma\) is cofinite (the result holds on split graphs).
  \end{enumerate}
\end{enumerate}
The same conclusions hold when \(w=\lrw(G)\), when \(w=\rw(G)\), or when
\(w\) is the width of a supplied rank-decomposition.
Under \(\#\mathrm{ETH}\), the same lower bounds hold for counting the solutions of size at most (or exactly) the target. In fact, all reductions are parsimonious at the target size.
\end{theorem}

\paragraph{Unweighted Dominating Set (Section~\ref{sec:unweighted-ds}).}
This resolves the explicit open problem of Bergougnoux, Korhonen, and Nederlof~\cite{bergougnoux-korhonen-nederlof-2023}: weights do not account for the quadratic exponent.  Combined with the known upper bounds~\cite{bui-xuan-telle-vatshelle-2010,bui-xuan-telle-vatshelle-2013}, our result establishes, under ETH, that the quadratic rank-width exponent is tight.  The lower bound already holds on split graphs and on bipartite graphs of diameter at most four, so neither a dense split structure nor a bipartite structure of small diameter removes the obstruction.  The reduction additionally gives the corresponding \(\#\mathrm{ETH}\) lower bounds.

\paragraph{Independent, connected, and total domination
	(Section~\ref{sec:three-variants}).}
Section~\ref{sec:three-variants} extends the reduction to Independent, Connected, and Total Dominating Set.  In each case, the target budget forces every solution of size at most the target to encode a unique satisfying assignment.  The original monopolar construction yields the lower bound for Independent Dominating Set.  Modified constructions yield the lower bounds for Connected and Total Dominating Set on split graphs and, separately, on bipartite graphs of diameter at most four.  These modifications add adjacencies among the selected vertices, which ensures connectivity and total domination but violates independence.  This is why the Independent Dominating Set result is stated only for monopolar graphs. The resulting lower bounds match the known \(2^{O(\rw(G)^2)}\) upper bounds for Independent and Total Dominating Set~\cite{bui-xuan-telle-vatshelle-2013} and for Connected Dominating Set~\cite{bergougnoux-kante-2021}. Consequently, the quadratic rank-width exponent is tight for all three variants on the stated graph classes.

\paragraph{\((\sigma,\rho)\)-sets with cofinite \(\rho\)
(Section~\ref{sec:sigma-rho}).}
Section~\ref{sec:sigma-rho} extends the quadratic lower bound to fixed pairs \((\sigma,\rho)\) for which \(\rho\) is cofinite and \(0\notin\rho\).  We consider two cases.  If \(0\in\sigma\) and \(1\in\rho\), the lower bound holds on monopolar graphs, without any finiteness assumption on \(\sigma\).  If \(\sigma\) is cofinite, the lower bound holds on split graphs.  Consequently, among pairs in which both sets are finite or cofinite, the theorem covers every cofinite--cofinite pair with \(0\notin\rho\), as well as every finite--cofinite pair satisfying \(0\in\sigma\) and \(1\in\rho\).  For every such finite/cofinite pair, the decision lower bound matches the general \(2^{O(\rw(G)^2)}\) upper bound in \eqref{eq:general-rw-upper-bound}.  Table~\ref{tab:sigma-rho-coverage} records the covered problems and the cases that remain open.
In particular, this settles the entire cofinite--cofinite minimization regime: when \(0\in\rho\), the empty set is feasible, and when \(0\notin\rho\), our lower bound matches the general upper bound.

\paragraph{Quantitative counting lower bounds (Appendix~\ref{app:seth-constants}).}
Under \(\#\mathrm{SETH}\), the basic monopolar construction excludes
\(2^{(1/9-\varepsilon)w^2}n^{O(1)}\) time with a supplied
rank-decomposition for every \(\varepsilon\in(0,1/9)\), and excludes
\(2^{(1/16-\varepsilon)w^2}n^{O(1)}\) time with a supplied vertex order
for every \(\varepsilon\in(0,1/16)\).

\paragraph{Perfect Code on split graphs (Appendix~\ref{app:perfect-code}).}
As a complementary algorithmic result, we completely characterize perfect
codes on split graphs.  This yields, in linear time, an implicit
representation and count of all perfect codes, together with minimum- and
maximum-cardinality solutions.

\begin{table}[!b]
	\centering
	\small
	\begin{tabular}{@{}L{0.30\textwidth}L{0.24\textwidth}
			L{0.24\textwidth}c@{}}
		\toprule
		Problem & \(\sigma\) & \(\rho\) & Coverage \\
		\midrule
		Dominating Set
		& \(\mathbb N_{\ge 0}\)
		& \(\mathbb N_{\geq1}\)
		& \(\checkmark\,\star\) \\
		
		Independent Dominating Set
		& \(\{0\}\)
		& \(\mathbb N_{\geq1}\)
		& \(\checkmark\,\star\) \\
		
		Connected Dominating Set
		& --- & ---
		& \(\diamond\) \\
		
		Total Dominating Set
		& \(\mathbb N_{\geq1}\)
		& \(\mathbb N_{\geq1}\)
		& \(\checkmark\,\star\) \\
		
		\(q\)-Tuple Dominating Set, fixed \(q\geq1\)
		& \(\mathbb N_{\geq q-1}\)
		& \(\mathbb N_{\geq q}\)
		& \(\checkmark\) \\
	
		Perfect Dominating Set
		& \(\mathbb N_{\ge 0}\)
		& \(\{1\}\)
		& \(\circ\) \\
		
		Perfect Code
		& \(\{0\}\)
		& \(\{1\}\)
		& \(\circ\) \\
		\bottomrule
	\end{tabular}
	\caption{Specializations and limits of our results for cofinite
		\(\rho\).  Here, \(\checkmark\) means that the problem is covered by the
		general \((\sigma,\rho)\)-set theorem in
		Section~\ref{sec:sigma-rho}; \(\star\) indicates an additional direct
		treatment in Section~\ref{sec:unweighted-ds} or
		Section~\ref{sec:three-variants}; \(\diamond\) denotes a problem handled
		separately in Section~\ref{sec:three-variants} because it is not expressible by a fixed pair
		\((\sigma,\rho)\); and \(\circ\) denotes a case outside the present
		lower-bound theorems and discussed in Section~\ref{sec:conclusion}.
		For \(q\)-Tuple Dominating Set, we use the convention
		\(|N[v]\cap S|\geq q\) for every vertex \(v\).}
	\label{tab:sigma-rho-coverage}
\end{table}

\subsection{Proof ideas and organization}

The reduction begins with a 3-CNF formula on \(k^2\) variables and views an assignment as a \(k\times k\) Boolean matrix.  For each clause \(C_h\), the construction creates a layer containing \(k\) choice groups, one for each row of the assignment matrix.  Each group contains one vertex for every possible \(k\)-bit assignment to that row.  The tight budget forces a dominating set to choose exactly one vertex from every group, so the \(k\) selected vertices encode a complete \(k\times k\) assignment matrix \(X_h\). 

Two nonadjacent guards per choice group replace the prohibitive weights of the earlier weighted reduction.  Under a tight budget, omitting the assignment vertex forces both guards to be selected, so every group contains exactly one selected assignment vertex.  This rigidity is also what makes the reduction parsimonious.

The central equality checker synchronizes two consecutive matrices: a disagreement is witnessed by a vector \(t\in\F^k\), yet each side of the checker's adjacency interface has rank at most \(2k\).  Thus the construction transports \(k^2\) bits of assignment information through cuts of rank \(O(k)\).

Completing the assignment vertices to a clique yields split instances, whereas deleting the choice-clique edges and adding a forced hub yields bipartite instances. Each modification increases the relevant cut-rank by only a constant. Finally, Section~\ref{sec:sigma-rho} adds a constant number of forced selected vertices to adjust the selected-neighbor counts of the clause vertices and equality checkers.  This converts the requirement of having at least one selected assignment neighbor into a \(\rho\)-constraint. The forced vertices form an independent set in the \(0\in\sigma,1\in\rho\) construction and a clique when \(\sigma\) is cofinite.

Section~\ref{sec:preliminaries} fixes notation, width measures, counting problems, and the ETH assumptions.  Section~\ref{sec:unweighted-ds} proves the basic theorem and its graph-class and counting refinements. Section~\ref{sec:three-variants} treats independent, connected, and total domination.  Section~\ref{sec:sigma-rho} proves the \((\sigma,\rho)\)-set framework theorem for two families of pairs. Section~\ref{sec:conclusion} concludes with the remaining finite/cofinite regimes in the \((\sigma,\rho)\)-set framework.

\section{Preliminaries}
\label{sec:preliminaries}

\subsection{Basic notation and graph classes}

For every integer \(s\geq0\), write \([s]:=\{1,\ldots,s\}\), with \([0]:=\varnothing\), and let \(\mathbb N_{\ge 0}:=\{0,1,2,\ldots\}\).  All graphs are finite, undirected, and simple.  For a graph \(G\), its vertex and edge sets are \(V(G)\) and \(E(G)\), and \(n:=|V(G)|\).  For \(v\in V(G)\), let \(N_G(v)\) and \(N_G[v]:=N_G(v)\cup\{v\}\) denote its open and closed neighborhoods.  For \(X\subseteq V(G)\), write \(\overline X:=V(G)\setminus X\) and \(G[X]\) for the subgraph induced by \(X\).  Subscripts are omitted when the graph is clear.  The distance between two vertices is the length of a shortest path, and the diameter is the maximum distance between two vertices in a connected graph.

A \emph{cluster graph} is a disjoint union of cliques.  A graph is \emph{monopolar} if its vertices can be partitioned into an independent set and a cluster graph.  It is a \emph{split graph} if its vertices can be partitioned into one clique and one independent set.  A graph is \emph{bipartite} if its vertices can be partitioned into two independent sets.

\subsection{Domination problems and their counting versions}

A set \(D\subseteq V(G)\) dominates \(G\) if \(N[v]\cap D\neq\varnothing\) for every \(v\in V(G)\).  A dominating set is \emph{independent} if \(G[D]\) has no edge and \emph{connected} if \(G[D]\) is connected.  It is a \emph{total dominating set} if \(N(v)\cap D\neq\varnothing\) for every \(v\in V(G)\), including vertices of \(D\) itself.  For a fixed integer \(q\geq1\), a \(q\)-tuple dominating set satisfies \(|N[v]\cap D|\geq q\) for every vertex.  Formula \eqref{eq:sigma-rho-definition} formally defines a \((\sigma,\rho)\)-set; throughout Section~\ref{sec:sigma-rho}, the pair is fixed and only the graph and target size belong to the input.

For ordinary domination, define
\begin{align*}
  \Zle(G,\ell)
  &:=\bigl|\{D\subseteq V(G):D\text{ dominates }G,
       |D|\leq\ell\}\bigr|,\\
  \Zeq(G,\ell)
  &:=\bigl|\{D\subseteq V(G):D\text{ dominates }G,
       |D|=\ell\}\bigr|.
\end{align*}
For a fixed pair \((\sigma,\rho)\), define analogously
\begin{align*}
  Z^{\sigma,\rho}_{\leq}(G,\ell)
  &:=\bigl|\{S\subseteq V(G):S\text{ is a }(\sigma,\rho)\text{-set},
       |S|\leq\ell\}\bigr|,\\
  Z^{\sigma,\rho}_{=}(G,\ell)
  &:=\bigl|\{S\subseteq V(G):S\text{ is a }(\sigma,\rho)\text{-set},
       |S|=\ell\}\bigr|.
\end{align*}
The variant-specific symbols in Section~\ref{sec:three-variants} have the same meaning with the corresponding additional property.  We count only solutions up to, or exactly at, the input target; we do not claim that the reductions control all larger solutions.  A reduction is \emph{parsimonious} if it induces a bijection between source and target solutions, and hence preserves their number exactly.

\subsection{Cut-rank, rank-width, and linear rank-width}

Let \(A_G\) be the adjacency matrix of \(G\).  For \(X\subseteq V(G)\), the \emph{cut-rank} of \(X\) is
\begin{equation}
  \rho_G(X):=
  \operatorname{rank}_{\mathbb F_2}
  \bigl(A_G[X,\overline X]\bigr).
  \label{eq:cut-rank}
\end{equation}
It is symmetric: \(\rho_G(X)=\rho_G(\overline X)\).

A \emph{rank-decomposition} of a graph \(G\) with at least two vertices is a pair \((T,\delta)\), where \(T\) is a subcubic tree and \(\delta\) is a bijection from \(V(G)\) to the leaves of \(T\).  Every edge \(e\in E(T)\) partitions the leaves, and therefore \(V(G)\), into sets \(X_e,\overline{X_e}\).  The width of \((T,\delta)\) is
\[
  \max_{e\in E(T)}\rho_G(X_e),
\]
and the \emph{rank-width} \(\rw(G)\) is the minimum width of a rank-decomposition of \(G\)~\cite{oum-seymour-2006}.  We set \(\rw(G)=0\) when \(|V(G)|<2\).

For a vertex order \(\pi=(v_1,\ldots,v_n)\), define
\begin{equation}
  \width_G(\pi):=
  \max_{1\leq i<n}
  \rho_G(\{v_1,\ldots,v_i\}).
  \label{eq:layout-width}
\end{equation}
For \(|V(G)|<2\), we set \(\width_G(\pi)=0\).
The \emph{linear rank-width} is \(\lrw(G):=\min_\pi\width_G(\pi)\), with \(\lrw(G)=0\) when \(|V(G)|<2\).  A vertex order can be converted in polynomial time into a caterpillar rank-decomposition of width at most \(\width_G(\pi)\), and in particular
\begin{equation}
  \rw(G)\leq\lrw(G)\leq\width_G(\pi).
  \label{eq:width-chain}
\end{equation}
All ranks in this paper, including those used to certify the supplied orders, are over \(\mathbb F_2\).

\subsection{ETH, \#ETH, and the source problem}

The Exponential Time Hypothesis (ETH) states that satisfiability of a 3-CNF formula with \(N\) variables and \(m\) clauses cannot be decided in \(2^{o(N)}(N+m)^{O(1)}\) time~\cite{impagliazzo-paturi-2001}.  The sparsification lemma makes this decision formulation robust with respect to the number of clauses~\cite{impagliazzo-paturi-zane-2001}.  Its counting analogue \(\#\mathrm{ETH}\) states the corresponding lower bound for counting satisfying assignments~\cite{dell-et-al-2014}; disjoint counting sparsification yields the corresponding clause-robust form~\cite[Theorem~1.1 and Appendix~A]{dell-et-al-2014}.

We use the following square-variable consequence.

\begin{lemma}[Square-variable ETH and \(\#\mathrm{ETH}\)]
\label{lem:square-eth}
Unless ETH fails, 3-SAT on formulas with \(k^2\) variables and \(m\) clauses cannot be decided in time
\begin{equation}
  2^{o(k^2)}(k+m)^{O(1)}.
  \label{eq:square-eth}
\end{equation}
Assuming \(\#\mathrm{ETH}\), their satisfying assignments cannot be counted within the same time bound.
\end{lemma}

Indeed, let \(k=\lceil\sqrt N\rceil\), add \(k^2-N\) fresh variables, and for each fresh variable \(y\) add the unit clause \((\neg y)\).  Here a 3-CNF formula is allowed to contain clauses of size at most three.  Every satisfying assignment then has a unique extension, so the padding is parsimonious.  The padded formula has
\[
  m'=m+k^2-N=m+O(\sqrt N)
\]
clauses.  Since \(k^2=N+O(\sqrt N)\), either claimed algorithm would contradict the corresponding hypothesis.

\section{The unweighted Dominating Set lower bound}
\label{sec:unweighted-ds}

This section proves the basic reduction for unweighted Dominating Set. We give the construction and all rank calculations explicitly, because the same construction will be modified in later sections. We then record two graph-class refinements that require only local changes: one gives split graphs, and the other gives bipartite graphs of diameter at most four. Finally, we identify the precise counting problems for which the reduction is parsimonious.  We first state the formal conclusions proved in this section.

\begin{theorem}[Unweighted Dominating Set]
\label{thm:unweighted-ds}
Unless ETH fails, no algorithm can take an unweighted graph \(G\), an integer \(d\), and a vertex order \(\pi\), and decide whether \(\gamma(G)\leq d\) in time
\[
  2^{o(\width_G(\pi)^2)}|V(G)|^{O(1)}.
\]
The lower bound holds separately on monopolar graphs, split graphs, and bipartite graphs of diameter at most four.  For formulas with \(k^2\) variables, the respective constructions supply orders of width at most \(4k+2\), \(4k+3\), and \(4k+3\).  The conclusion also holds when the parameter is \(\lrw(G)\), \(\rw(G)\), or the width of a supplied rank-decomposition.
\end{theorem}

\begin{theorem}[Parsimonious counting consequence]
\label{thm:counting-ds}
Unless \(\#\mathrm{ETH}\) fails, no algorithm can take a graph \(G\), a target \(d\), and a vertex order \(\pi\) and, on any of the three graph classes in Theorem~\ref{thm:unweighted-ds}, compute either \(\Zle(G,d)\) or \(\Zeq(G,d)\) in time
\[
  2^{o(\width_G(\pi)^2)}|V(G)|^{O(1)}.
\]
The conclusion also holds when the parameter is linear rank-width, rank-width, or the width of a supplied rank-decomposition.
\end{theorem}

\subsection{Source problem and assignment matrices}

We apply Lemma~\ref{lem:square-eth}.  Let the input formula \(\varphi\) have \(k^2\) variables and clauses \(C_1,\ldots,C_m\).

Let the variables be
\[
  \{v_{a,b}:a,b\in[k]\}.
\]
An assignment is identified with a matrix \(X\in\F^{k\times k}\). Its \(a\)-th row \(x_a\in\F^k\) assigns \(v_{a,1},\ldots,v_{a,k}\). We may assume that every clause is nonempty and that \(k,m\geq1\), since all excluded cases are decidable in polynomial time.

\subsection{Construction}
\label{subsec:basic-construction}

From a formula \(\varphi\), construct a graph \(G_\varphi\) as follows.

\paragraph{Choice groups.}
For every \(h\in[m]\) and \(a\in[k]\), create
\[
  A_{h,a}:=\{a_{h,a,x}:x\in\F^k\}
\]
and make \(A_{h,a}\) a clique. Create two nonadjacent guards \(g^0_{h,a}\) and \(g^1_{h,a}\), each having open neighborhood exactly \(A_{h,a}\). There are no edges between distinct choice cliques.

\paragraph{Clause vertices.}
For every \(h\in[m]\), create a vertex \(c_h\). Join \(c_h\) to \(a_{h,a,x}\) if and only if \(x\) satisfies at least one literal of \(C_h\) whose variable belongs to row \(a\). In particular, a literal \(v_{a,b}\) is satisfied when \(x_b=1\), and a literal \(\neg v_{a,b}\) is satisfied when \(x_b=0\).

\paragraph{Equality checkers.}
For every \(h\in[m-1]\), \(t,p\in\F^k\), and \(r\in\F^k\setminus\{0\}\), create a vertex \(e_{h,t,p,r}\). All checkers are pairwise nonadjacent. Their only neighbors lie in the two endpoint layers and are specified by
\begin{align}
  e_{h,t,p,r}a_{h,a,x}\in E(G_\varphi)
  &\iff \langle x,t\rangle\neq p_a,
  \label{eq:checker-left}\\
  e_{h,t,p,r}a_{h+1,a,x}\in E(G_\varphi)
  &\iff \langle x,t\rangle\neq p_a+r_a.
  \label{eq:checker-right}
\end{align}
All arithmetic in these expressions is over \(\F\). No other edges are added. 

Figure~\ref{fig:basic-construction-example} illustrates the construction
for a formula with four variables and two clauses.  Since this
small example already has \(48\) equality checkers, the figure
expands only one representative checker.

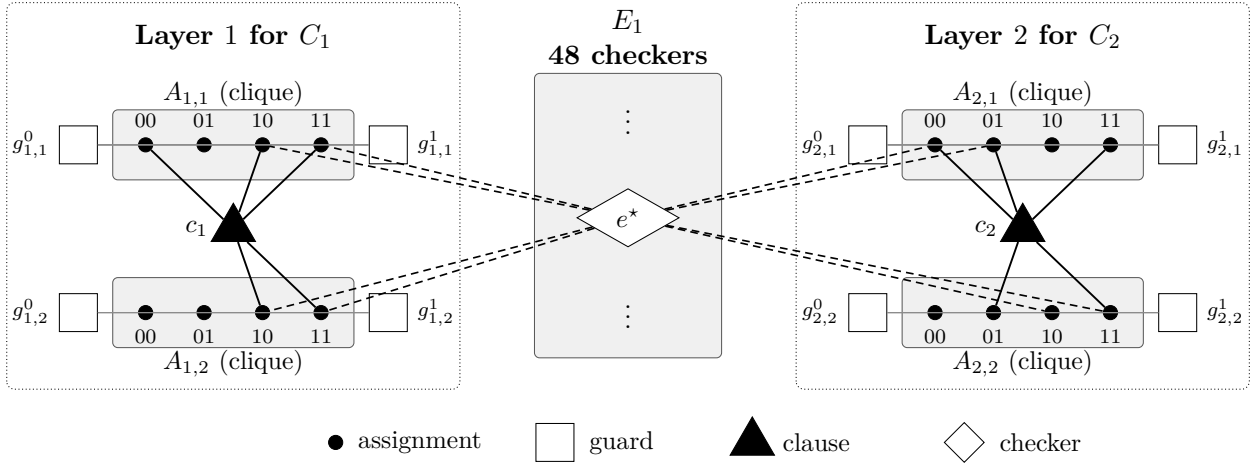
\begin{figure}[htbp]
	\centering
	\resizebox{\textwidth}{!}{%
		\begin{tikzpicture}[
			font=\small,
			assignment/.style={
				circle,draw=black,fill=black,
				minimum size=5.5pt,inner sep=0pt
			},
			guard/.style={
				rectangle,draw=black,fill=white,
				minimum size=5.2mm,inner sep=0pt
			},
			clause/.style={
				regular polygon,regular polygon sides=3,
				draw=black,fill=black,
				minimum size=7mm,inner sep=0pt
			},
			checker/.style={
				diamond,aspect=1.45,draw=black,fill=white,
				minimum width=14mm,minimum height=8mm,inner sep=1pt
			},
			group/.style={
				rounded corners=2pt,draw=black!65,fill=black!6
			},
			layer/.style={
				rounded corners=3pt,draw=black,densely dotted
			},
			guardedge/.style={
				draw=black!45,line width=.35pt
			},
			clauseedge/.style={
				draw=black,line width=.75pt
			},
			checkeredge/.style={
				draw=black,densely dashed,line width=.65pt
			}
			]
			
			% Layer and choice-group boundaries.
			\draw[layer] (-0.50,-1.05) rectangle (5.70,4.25);
			\draw[layer] (10.30,-1.05) rectangle (16.50,4.25);
			
			\draw[group] (0.95,1.82) rectangle (4.25,2.78);
			\draw[group] (0.95,-.48) rectangle (4.25,.48);
			\draw[group] (11.75,1.82) rectangle (15.05,2.78);
			\draw[group] (11.75,-.48) rectangle (15.05,.48);
			
			% The equality-checker block.
			\draw[group] (6.72,-.62) rectangle (9.28,3.28);
			
			% Assignment vertices in layer 1.
			\foreach \s/\x in {00/1.40,01/2.20,10/3.00,11/3.80}{
				\node[assignment,label=above:{\scriptsize\(\s\)}]
				(a11\s) at (\x,2.30) {};
				\node[assignment,label=below:{\scriptsize\(\s\)}]
				(a12\s) at (\x,0.00) {};
			}
			
			% Assignment vertices in layer 2.
			\foreach \s/\x in {00/12.20,01/13.00,10/13.80,11/14.60}{
				\node[assignment,label=above:{\scriptsize\(\s\)}]
				(a21\s) at (\x,2.30) {};
				\node[assignment,label=below:{\scriptsize\(\s\)}]
				(a22\s) at (\x,0.00) {};
			}
			
			% Guard vertices.
			\node[guard,label=left:{\scriptsize\(g^0_{1,1}\)}]
			(g110) at (.48,2.30) {};
			\node[guard,label=right:{\scriptsize\(g^1_{1,1}\)}]
			(g111) at (4.72,2.30) {};
			
			\node[guard,label=left:{\scriptsize\(g^0_{1,2}\)}]
			(g120) at (.48,0.00) {};
			\node[guard,label=right:{\scriptsize\(g^1_{1,2}\)}]
			(g121) at (4.72,0.00) {};
			
			\node[guard,label=left:{\scriptsize\(g^0_{2,1}\)}]
			(g210) at (11.28,2.30) {};
			\node[guard,label=right:{\scriptsize\(g^1_{2,1}\)}]
			(g211) at (15.52,2.30) {};
			
			\node[guard,label=left:{\scriptsize\(g^0_{2,2}\)}]
			(g220) at (11.28,0.00) {};
			\node[guard,label=right:{\scriptsize\(g^1_{2,2}\)}]
			(g221) at (15.52,0.00) {};
			
			% Each guard is adjacent to every assignment vertex in its group.
			\foreach \s in {00,01,10,11}{
				\draw[guardedge] (g110)--(a11\s);
				\draw[guardedge] (g111)--(a11\s);
				\draw[guardedge] (g120)--(a12\s);
				\draw[guardedge] (g121)--(a12\s);
				
				\draw[guardedge] (g210)--(a21\s);
				\draw[guardedge] (g211)--(a21\s);
				\draw[guardedge] (g220)--(a22\s);
				\draw[guardedge] (g221)--(a22\s);
			}
			
		% Clause vertices placed between the two assignment cliques.
		\node[clause,label=left:{\(c_1\)}] (c1) at (2.60,1.15) {};
		\node[clause,label=left:{\(c_2\)}] (c2) at (13.40,1.15) {};
			
			% C_1=(v_{1,1} or not v_{1,2} or v_{2,1}).
			\foreach \v in {a1100,a1110,a1111,a1210,a1211}{
				\draw[clauseedge] (c1)--(\v);
			}
			
			% C_2=(not v_{1,1} or v_{1,2} or v_{2,2}).
			\foreach \v in {a2100,a2101,a2111,a2201,a2211}{
				\draw[clauseedge] (c2)--(\v);
			}
			
			% One representative equality checker:
			% t=(1,0), p=(0,0), r=(1,0).
			\node at (8.00,2.72) {\(\vdots\)};
			\node[checker] (estar) at (8.00,1.30) {\(e^\star\)};
			\node at (8.00,.05) {\(\vdots\)};
			
		\foreach \v in {a1110,a1111,a1210,a1211,a2100,a2101,a2210,a2211}{%
			\draw[checkeredge] (estar)--(\v);
		}
			% Labels.
			\node[font=\bfseries] at (2.60,3.78)
			{Layer \(1\) for \(C_1\)};
			\node[font=\bfseries] at (13.40,3.78)
			{Layer \(2\) for \(C_2\)};
			
			\node[font=\bfseries,align=center] at (8.00,3.78)
			{\(E_1\)\\48 checkers};
			
		    \node[above] at (2.60,2.68) {\(A_{1,1}\) (clique)};
			\node[below] at (2.60,-.38) {\(A_{1,2}\) (clique)};
			\node[above] at (13.40,2.68) {\(A_{2,1}\) (clique)};
			\node[below] at (13.40,-.38) {\(A_{2,2}\) (clique)};
			
			% Vertex-type legend.
			\node[assignment] (la) at (4.00,-1.78) {};
			\node[right=2pt of la] {assignment};
			
			\node[guard] (lg) at (7.00,-1.78) {};
			\node[right=2pt of lg] {guard};
			
			\node[clause] (lc) at (9.70,-1.78) {};
			\node[right=2pt of lc] {clause};
			
			\node[
			diamond,aspect=1.3,draw=black,fill=white,
			minimum size=5.5mm,inner sep=0pt
			] (le) at (12.60,-1.78) {};
			\node[right=2pt of le] {checker};
			
		\end{tikzpicture}%
	}
	\caption{An example of the construction for \(k=2\) and \(m=2\).
		The variables are \(v_{1,1},v_{1,2},v_{2,1},v_{2,2}\), and
		\(\varphi=(v_{1,1}\lor\neg v_{1,2}\lor v_{2,1})
		\land(\neg v_{1,1}\lor v_{1,2}\lor v_{2,2})\).
		The labels \(00,01,10,11\) denote the four assignments
		\(x\in\F^2\) to a row. Each shaded choice group is a clique, although
		its internal clique edges are suppressed. Each square guard is adjacent
		to all four assignment vertices of its group, as indicated by the thin
		gray edges. Solid edges incident with \(c_h\) represent assignments
		satisfying at least one literal of \(C_h\). The block \(E_1\) contains
		the \(48\) equality checkers between the two layers. Only
		\(e^\star=e_{1,(1,0),(0,0),(1,0)}\) and its dashed incident edges are
		drawn; the remaining checker vertices and their incident edges are
		suppressed.}
	\label{fig:basic-construction-example}
\end{figure}
The number of vertices is
\begin{align*}
  |V(G_\varphi)|
  &=mk2^k+2mk+m+(m-1)2^{2k}(2^k-1)\\
  &=m2^{O(k)}.
\end{align*}
Every checker has possible neighbors only among the \(2k2^k\) assignment vertices of its two endpoint layers. Consequently, the graph and the order used below can be generated explicitly in \(m2^{O(k)}\operatorname{poly}(k)\) time.

The assignment vertices induce a disjoint union of cliques, while the guards, clause vertices, and checkers form an independent set. Thus \(G_\varphi\) is monopolar: its vertex set can be partitioned into a cluster graph and an independent set.

Set
\begin{equation}
  d:=km.
  \label{eq:basic-target}
\end{equation}

\subsection{Tight-budget normal form}

For \((h,a)\in[m]\times[k]\), put
\[
  Q_{h,a}:=A_{h,a}\cup\{g^0_{h,a},g^1_{h,a}\}.
\]

\begin{lemma}[Tight-budget normal form]
\label{lem:basic-normal-form}
Every dominating set of \(G_\varphi\) has size at least \(km\). If \(D\) is a dominating set with \(|D|\leq km\), then
\begin{enumerate}[label=\textup{(\roman*)}]
  \item \(|D|=km\);
  \item \(D\) contains exactly one vertex of every \(A_{h,a}\); and
  \item \(D\) contains no guard, clause vertex, or equality checker.
\end{enumerate}
\end{lemma}

\begin{proof}
Fix \((h,a)\). If \(D\cap A_{h,a}=\varnothing\), then both guards must belong to \(D\): the guards are nonadjacent, neither has a neighbor outside \(A_{h,a}\), and hence neither can otherwise be dominated. If \(D\cap A_{h,a}\neq\varnothing\), then \(D\cap Q_{h,a}\neq\varnothing\). In either case,
\[
  |D\cap Q_{h,a}|\geq1.
\]
The \(km\) sets \(Q_{h,a}\) are pairwise disjoint, so \(|D|\geq km\).

Suppose now that \(|D|\leq km\). Equality must hold in every one of the preceding \(km\) inequalities, and no vertex outside \(\bigcup_{h,a}Q_{h,a}\) can belong to \(D\). For a fixed pair \((h,a)\), the unique vertex of \(D\cap Q_{h,a}\) cannot be a guard: if \(D\cap A_{h,a}=\varnothing\), both guards, rather than only one, are necessary. Hence \(D\cap A_{h,a}\) consists of exactly one assignment vertex. This also proves that no guard belongs to \(D\), completing the proof.
\end{proof}

Whenever Lemma~\ref{lem:basic-normal-form} applies, write the unique vertex of \(D\cap A_{h,a}\) as \(a_{h,a,x_{h,a}}\), and let \(X_h\in\F^{k\times k}\) be the matrix whose \(a\)-th row is \(x_{h,a}\).

\subsection{The equality checker}

\begin{lemma}[Exact equality test]
\label{lem:basic-equality}
Suppose the selected assignment vertices in two consecutive layers \(h\) and \(h+1\) encode matrices \(X,Y\in\F^{k\times k}\). All checkers between these layers have a selected assignment neighbor if and only if \(X=Y\).
\end{lemma}

\begin{proof}
Fix a checker \(e_{h,t,p,r}\). By \eqref{eq:checker-left}, it has no selected neighbor in layer \(h\) if and only if
\[
  \langle x_a,t\rangle=p_a
  \quad\text{for every }a\in[k],
\]
which is equivalent to \(Xt=p\). Similarly, by \eqref{eq:checker-right}, it has no selected neighbor in layer \(h+1\) if and only if \(Yt=p+r\). Therefore
\begin{equation}
  e_{h,t,p,r}\text{ is not dominated by the selected assignment vertices}
  \iff Xt=p\ \text{ and }\ Yt=p+r.
  \label{eq:undominated-checker}
\end{equation}

If \(X=Y\), the two equations on the right-hand side of \eqref{eq:undominated-checker} imply \(r=0\), contrary to the definition of a checker. Thus every checker has a selected neighbor.

Conversely, suppose \(X\neq Y\). The linear map \(X+Y\) is nonzero, so there exists \(t\in\F^k\) for which
\[
  r:=(X+Y)t\neq0.
\]
Set \(p:=Xt\). Since the field has characteristic two,
\[
  p+r=Xt+(X+Y)t=Yt.
\]
The checker \(e_{h,t,p,r}\) therefore satisfies the right-hand side of \eqref{eq:undominated-checker} and has no selected neighbor. This proves the converse.
\end{proof}

Appendix~\ref{app:standard-checker} gives a smaller standard-basis version of this equality checker with the same correctness and rank guarantees.

\subsection{Correctness and parsimony}

For \(X\in\F^{k\times k}\), with rows \(x_1,\ldots,x_k\), define
\begin{equation}
  D_X:=\{a_{h,a,x_a}:h\in[m],\ a\in[k]\}.
  \label{eq:canonical-ds}
\end{equation}

\begin{lemma}[Bijection with satisfying assignments]
\label{lem:basic-bijection}
The map \(X\mapsto D_X\) is a bijection from the satisfying assignments of \(\varphi\) to the dominating sets of \(G_\varphi\) of size at most \(d\). Every such dominating set has size exactly \(d\). Consequently,
\begin{equation}
  \Zle(G_\varphi,d)
  =\Zeq(G_\varphi,d)
  =\#\operatorname{SAT}(\varphi).
  \label{eq:parsimonious-basic}
\end{equation}
In particular, \(\varphi\) is satisfiable if and only if \(\gamma(G_\varphi)\leq d\).
\end{lemma}

\begin{proof}
Let \(X\) satisfy \(\varphi\). The set \(D_X\) has size \(km=d\). In each choice group its selected vertex dominates the whole clique and both guards. Since \(X\) satisfies \(C_h\), at least one selected assignment vertex in layer \(h\) is adjacent to \(c_h\). Consecutive layers encode the same matrix \(X\), so Lemma~\ref{lem:basic-equality} shows that every equality checker is dominated. Thus \(D_X\) dominates \(G_\varphi\).

Conversely, let \(D\) dominate \(G_\varphi\) and satisfy \(|D|\leq d\). Lemma~\ref{lem:basic-normal-form} gives one matrix \(X_h\) in every layer and shows that neither clause vertices nor checkers belong to \(D\). Hence \(c_h\) must have a selected assignment neighbor, so \(X_h\) satisfies \(C_h\). Since every checker must also have a selected neighbor, Lemma~\ref{lem:basic-equality} gives
\[
  X_1=X_2=\cdots=X_m.
\]
The common matrix \(X\) satisfies every clause, and the normal form implies \(D=D_X\). The encoded matrix is uniquely determined by \(D\), so the two maps are inverse bijections. Equation~\eqref{eq:parsimonious-basic} follows.
\end{proof}

\subsection{A supplied order of width at most \texorpdfstring{\(4k+2\)}{4k+2}}

For \(h\in[m]\), let \(\Lambda_h\) be the following order of the vertices belonging to clause layer \(h\):
\[
  c_h,
  g^0_{h,1},g^1_{h,1},A_{h,1},
  \ldots,
  g^0_{h,k},g^1_{h,k},A_{h,k},
\]
where the vertices inside each \(A_{h,a}\) are ordered arbitrarily. Let
\[
  E_h:=\{e_{h,t,p,r}:t,p\in\F^k,
  \ r\in\F^k\setminus\{0\}\}
\]
be ordered arbitrarily. Define
\begin{equation}
  \pi:=\Lambda_1,E_1,\Lambda_2,E_2,\ldots,E_{m-1},\Lambda_m.
  \label{eq:basic-order}
\end{equation}

We first isolate the only two nonconstant-rank interfaces. Let \(M_L\) have rows indexed by \((a,x)\in[k]\times\F^k\) and columns indexed by
\[
  (t,p,r)\in\F^k\times\F^k\times(\F^k\setminus\{0\}).
\]
Its entries are
\[
  M_L[(a,x),(t,p,r)]
  =\langle x,t\rangle+p_a.
\]
Viewing \(t_i,p_i,r_i\) as functions of the column index, we have
\begin{equation}
  \operatorname{rowspan}(M_L)
  \subseteq\operatorname{span}
  \{t_1,\ldots,t_k,p_1,\ldots,p_k\},
  \qquad
  \rk(M_L)\leq2k.
  \label{eq:left-rank}
\end{equation}
For the right interface, the entries are
\[
  M_R[(a,x),(t,p,r)]
  =\langle x,t\rangle+p_a+r_a.
\]
Similarly,
\begin{equation}
  \operatorname{rowspan}(M_R)
  \subseteq\operatorname{span}
  \{t_1,\ldots,t_k,p_1+r_1,\ldots,p_k+r_k\},
  \qquad
  \rk(M_R)\leq2k.
  \label{eq:right-rank}
\end{equation}
The same bounds hold after deleting any rows or columns.

\begin{lemma}[Width of the supplied order]
\label{lem:basic-layout}
The order \(\pi\) in \eqref{eq:basic-order} satisfies
\[
  \width_{G_\varphi}(\pi)\leq4k+2.
\]
\end{lemma}

\begin{proof}
Consider a prefix cut of \(\pi\).

Suppose first that the cut lies inside a checker block \(E_h\). The only crossing edges are from assignment vertices in \(\Lambda_h\) to the unprocessed part of \(E_h\), and from the processed part of \(E_h\) to assignment vertices in \(\Lambda_{h+1}\). The two crossing submatrices are restrictions of \(M_L\) and \(M_R\). Rank subadditivity and \eqref{eq:left-rank}--\eqref{eq:right-rank} give cut-rank at most \(4k\).

Suppose next that the cut lies inside \(\Lambda_h\). Edges from \(E_{h-1}\) to unprocessed assignment vertices in \(\Lambda_h\) have rank at most \(2k\), and edges from processed assignment vertices in \(\Lambda_h\) to \(E_h\) have rank at most \(2k\). A missing checker block at either end contributes zero.

It remains to account for edges internal to \(\Lambda_h\). Every choice group is contiguous, so at most one choice group is split. The guards precede its assignment clique. If a cut separates the two guards, the single crossing guard row is constant on the clique. If the cut lies in the assignment clique, all processed guards and assignment vertices have the same all-one pattern toward the unprocessed part of that clique. Thus, in every case, the crossing edges internal to the split choice group have rank at most one. There are no edges between different choice groups. Finally, the single clause vertex \(c_h\) contributes at most one additional crossing row. The cut-rank is therefore at most
\[
  2k+2k+1+1=4k+2.
\]
These cases cover every prefix cut.
\end{proof}

For rank-width, as opposed to linear rank-width, Appendix~\ref{app:rank-decomposition} constructs rank-decompositions of width at most \(3k+O(1)\), including for the split and bipartite refinements that use the smaller checker of Appendix~\ref{app:standard-checker}.

\subsection{Split-graph instances}
\label{subsec:split-refinement}

Let \(G_\varphi^{\mathrm{sp}}\) be obtained from \(G_\varphi\) by adding all missing edges between assignment vertices. Thus
\[
  K:=\bigcup_{h\in[m],\,a\in[k]}A_{h,a}
\]
is a clique, while \(V(G_\varphi^{\mathrm{sp}})\setminus K\) is an independent set. Hence \(G_\varphi^{\mathrm{sp}}\) is a split graph.

\begin{proposition}[Split completion]
\label{prop:split-completion}
The map \(X\mapsto D_X\) is a bijection from the satisfying assignments of \(\varphi\) to the dominating sets of \(G_\varphi^{\mathrm{sp}}\) of size at most \(d\). Moreover, the order \(\pi\) satisfies
\[
  \width_{G_\varphi^{\mathrm{sp}}}(\pi)\leq4k+3.
\]
\end{proposition}

\begin{proof}
The added edges do not change the neighborhood of any guard, clause vertex, or equality checker. In particular, the two guards belonging to \(A_{h,a}\) still have neighborhood exactly \(A_{h,a}\). The proof of Lemma~\ref{lem:basic-normal-form} therefore applies verbatim. Under that normal form, every assignment vertex is dominated because all assignment vertices form a clique, while domination of clause vertices and checkers is governed by exactly the same conditions as before. The proof of Lemma~\ref{lem:basic-bijection} consequently remains valid.

For completeness, we bound the width directly; this avoids any assumption that completing a collection of cliques is a rank-one perturbation on an arbitrary cut. If a prefix cut lies inside \(E_h\), the two checker interfaces contribute at most \(4k\), and all crossing edges inside the global assignment clique form one complete bipartite graph, of rank at most one. The cut-rank is at most \(4k+1\).

If the cut lies inside \(\Lambda_h\), the two checker interfaces again contribute at most \(4k\). The global assignment clique contributes one all-one crossing matrix. At most one choice group is split, and its guard edges contribute rank at most one. The clause vertex contributes at most one more row. Thus the cut-rank is at most \(4k+3\).
\end{proof}

It follows in particular that
\begin{equation}
  \Zle(G_\varphi^{\mathrm{sp}},d)
  =\Zeq(G_\varphi^{\mathrm{sp}},d)
  =\#\operatorname{SAT}(\varphi).
  \label{eq:parsimonious-split}
\end{equation}

\subsection{Bipartite instances of diameter at most four}
\label{subsec:bipartite-refinement}

Construct \(G_\varphi^{\mathrm{bip}}\) as follows. Start with \(G_\varphi\), delete every edge whose two endpoints lie in the same choice clique \(A_{h,a}\), and add three new vertices \(z,z^0,z^1\). Join \(z\) to every assignment vertex and to \(z^0,z^1\); the vertices \(z^0,z^1\) have no other neighbors. Set
\[
  d_{\mathrm{bip}}:=km+1.
\]

The graph is bipartite with parts
\begin{align*}
  L&:=\Bigl(\bigcup_{h,a}A_{h,a}\Bigr)\cup\{z^0,z^1\},\\
  R&:=\{z\}\cup
  \{g^i_{h,a}:h\in[m],a\in[k],i\in\{0,1\}\}
  \cup\{c_h:h\in[m]\}
  \cup\bigcup_{h\in[m-1]}E_h.
\end{align*}

\begin{lemma}[Normal form with a forced hub]
\label{lem:bip-normal-form}
Every dominating set of \(G_\varphi^{\mathrm{bip}}\) has size at least \(km+1\). If \(D\) is a dominating set with \(|D|\leq km+1\), then
\[
  D=\{z\}\cup
  \{a_{h,a,x_{h,a}}:h\in[m],a\in[k]\}
\]
for uniquely determined vectors \(x_{h,a}\in\F^k\).
\end{lemma}

\begin{proof}
The argument for each \(Q_{h,a}\) does not use the clique edges inside \(A_{h,a}\): if \(D\cap A_{h,a}=\varnothing\), both guards must be selected, whereas otherwise \(D\cap Q_{h,a}\neq\varnothing\). Hence every \(Q_{h,a}\) contributes at least one vertex.

Let \(H:=\{z,z^0,z^1\}\). If \(z\notin D\), then both private leaves \(z^0,z^1\) must belong to \(D\); otherwise \(D\cap H\neq\varnothing\). Thus \(|D\cap H|\geq1\). The \(km\) choice sets and \(H\) are pairwise disjoint, proving \(|D|\geq km+1\).

At equality, every one of these \(km+1\) sets contributes exactly one vertex, and no vertex outside their union is selected. The unique vertex in \(H\) must be \(z\), because omitting \(z\) requires both leaves. The unique vertex in \(Q_{h,a}\) must lie in \(A_{h,a}\), because omitting \(A_{h,a}\) requires both guards. This proves the claimed form and its uniqueness.
\end{proof}

\begin{proposition}[Bipartite diameter-four refinement]
\label{prop:bip-refinement}
The following statements hold.
\begin{enumerate}[label=\textup{(\roman*)}]
  \item \(G_\varphi^{\mathrm{bip}}\) is bipartite and has diameter at most four.
  \item The map \(X\mapsto D_X\cup\{z\}\) is a bijection from the satisfying assignments of \(\varphi\) to the dominating sets of \(G_\varphi^{\mathrm{bip}}\) of size at most \(d_{\mathrm{bip}}\). Every such set has size exactly \(d_{\mathrm{bip}}\).
  \item The order
  \[
    \pi_{\mathrm{bip}}
    :=z^0,z^1,z,\Lambda_1,E_1,\ldots,E_{m-1},\Lambda_m,
  \]
  where the assignment sets are now independent, has width at most \(4k+3\).
\end{enumerate}
\end{proposition}

\begin{proof}
The displayed sets \(L,R\) form a bipartition. Every vertex of \(R\setminus\{z\}\) has an assignment neighbor. This is immediate for a guard. A clause vertex has an assignment neighbor because its clause is nonempty. Finally, a checker also has an assignment neighbor: if it has none in its earlier endpoint layer, then \(t=0\) and \(p=0\); since \(r\neq0\), some coordinate of \(p+r\) equals one, giving neighbors in its later endpoint layer. Every two vertices of \(L\) have the common neighbor \(z\). A vertex of \(R\setminus\{z\}\) is at distance two from \(z\), through an assignment neighbor. It follows that two vertices of \(R\) are at distance at most four, and a vertex of \(L\) and a vertex of \(R\) are at distance at most three. Hence the diameter is at most four.

Let \(X\) satisfy \(\varphi\). The hub \(z\) dominates all assignment vertices and its two leaves. The selected assignment vertices dominate all guards, clause vertices, and checkers exactly as in Lemma~\ref{lem:basic-bijection}. Therefore \(D_X\cup\{z\}\) is a dominating set of size \(km+1\).

Conversely, let \(D\) be a dominating set of size at most \(km+1\). Lemma~\ref{lem:bip-normal-form} gives the hub and one assignment vertex in every choice group, and excludes clause vertices and checkers from \(D\). The clause and equality arguments in Lemmas~\ref{lem:basic-equality} and \ref{lem:basic-bijection} are unchanged. Thus all layers encode one common satisfying matrix \(X\), and \(D=D_X\cup\{z\}\). This proves the bijection.

It remains to bound the width. A cut inside the initial three-vertex hub block has rank at most one. At every later cut, the crossing edges incident with \(z\) form a star and increase the rank by at most one. If a cut lies inside a checker block, the two checker interfaces contribute at most \(4k\), for a total of at most \(4k+1\). If it lies inside a layer, the two checker interfaces contribute at most \(4k\); the split guard--assignment incidence, the clause vertex, and the hub each increase the rank by at most one. The total is at most \(4k+3\).
\end{proof}

Consequently,
\begin{equation}
  \Zle(G_\varphi^{\mathrm{bip}},d_{\mathrm{bip}})
  =\Zeq(G_\varphi^{\mathrm{bip}},d_{\mathrm{bip}})
  =\#\operatorname{SAT}(\varphi).
  \label{eq:parsimonious-bip}
\end{equation}

\subsection{Decision and counting lower bounds}

We now prove the two theorems stated at the beginning of the section.  The phrase \emph{a supplied order of width \(w\)} means that the order \(\pi\) is part of the input and \(w=\width_G(\pi)\).

\begin{proof}[Proof of Theorem~\ref{thm:unweighted-ds}]
Given \(\varphi\) with \(k^2\) variables, use \(G_\varphi\), \(G_\varphi^{\mathrm{sp}}\), or \(G_\varphi^{\mathrm{bip}}\), with the corresponding target and supplied order. The relevant bijection proves that the resulting instance is a yes-instance exactly when \(\varphi\) is satisfiable. In every case the number of vertices is \(m2^{O(k)}\), the construction time is \(2^{O(k)}(k+m)^{O(1)}\), and the supplied width is \(O(k)\). An algorithm with the stated running time would therefore solve \(\varphi\) in \(2^{o(k^2)}m^{O(1)}\) time, contradicting \eqref{eq:square-eth}.
\end{proof}

\begin{corollary}[Rank-width formulation]
\label{cor:unweighted-rw}
Unless ETH fails, unweighted Dominating Set has no
\[
  2^{o(\rw(G)^2)}|V(G)|^{O(1)}
\]
algorithm, even on split graphs or on bipartite graphs of diameter at most four. The statement remains true if a rank-decomposition or the orders constructed above are supplied with the input.
\end{corollary}

\begin{proof}
For every constructed instance,
\[
  \rw(G)\leq\lrw(G)\leq\width_G(\pi)=O(k).
\]
Combining this inequality with the proof of Theorem~\ref{thm:unweighted-ds} gives the parameterized claims.  To obtain the supplied-decomposition formulation, convert \(\pi\) in polynomial time into the caterpillar rank-decomposition described in Section~\ref{sec:preliminaries}; its width is also \(O(k)\).
\end{proof}

\begin{proof}[Proof of Theorem~\ref{thm:counting-ds}]
Equations~\eqref{eq:parsimonious-basic}, \eqref{eq:parsimonious-split}, and \eqref{eq:parsimonious-bip} show that in each construction the number of solutions of size at most the target and the number of solutions of size exactly the target both equal \(\#\operatorname{SAT}(\varphi)\), with no multiplicative factor and no spurious solutions. The width, size, and construction-time calculation in Theorem~\ref{thm:unweighted-ds} therefore yields a \(2^{o(k^2)}m^{O(1)}\)-time algorithm for counting the satisfying assignments of \(\varphi\), contradicting \(\#\mathrm{ETH}\).  The linear-rank-width, rank-width, and supplied-decomposition formulations follow as in Corollary~\ref{cor:unweighted-rw}.
\end{proof}

Appendix~\ref{app:seth-constants} strengthens this qualitative counting lower bound for the basic monopolar construction under \(\#\mathrm{SETH}\), giving explicit constants for both the supplied-rank-decomposition and supplied-order formulations.

\begin{remark}[Meaning of parsimony]
The graph construction gives an actual bijection between satisfying assignments and target-size dominating sets. In particular, this is stronger than a reduction that merely preserves satisfiability. The restriction to \(\Zle\) or \(\Zeq\) is essential: dominating sets larger than the tight target generally exist and are intentionally not controlled by the reduction.
\end{remark}

\section{Independent, connected, and total domination}
\label{sec:three-variants}

The normal-form lemmas from Section~\ref{sec:unweighted-ds} do more than decide ordinary domination.  They identify \emph{every} dominating set up to the tight target.  We can therefore obtain three further lower bounds by checking which additional property is enjoyed by the canonical set in each of the three graph constructions.

\begin{theorem}[Three domination variants]
\label{thm:three-variants}
Unless ETH fails, none of the following problems admits an algorithm with running time
\[
  2^{o(\width_G(\pi)^2)}|V(G)|^{O(1)}
\]
when a vertex order \(\pi\) is supplied with the input:
\begin{enumerate}[label=\textup{(\roman*)}]
  \item \IDS{} on monopolar graphs;
  \item \CDS{} on split graphs, and separately on bipartite graphs of diameter at most four; and
  \item \TDS{} on split graphs, and separately on bipartite graphs of diameter at most four.
\end{enumerate}
For source formulas with \(k^2\) variables, the supplied vertex orders have width at most \(4k+2\) in (i), and at most \(4k+3\) in each split or bipartite construction in (ii)--(iii).  Every conclusion remains true when the parameter is \(\lrw(G)\), \(\rw(G)\), or the width of a supplied rank-decomposition.
\end{theorem}

\begin{theorem}[Counting the three variants]
\label{thm:counting-three-variants}
Assume \(\#\mathrm{ETH}\).  For each problem and graph class listed in Theorem~\ref{thm:three-variants}, no algorithm taking a graph \(G\), a target \(\ell\), and a vertex order \(\pi\) can compute either \(Z^{\mathcal P}_{\leq}(G,\ell)\) or \(Z^{\mathcal P}_{=}(G,\ell)\), for the corresponding \(\mathcal P\in\{\mathrm I,\mathrm C,\mathrm T\}\), in time
\[
  2^{o(\width_G(\pi)^2)}|V(G)|^{O(1)}.
\]
The conclusion also holds when the parameter is linear rank-width, rank-width, or the width of a supplied rank-decomposition.
\end{theorem}

We use the following standard definitions.  A dominating set \(D\) is \emph{independent} if \(G[D]\) has no edge and \emph{connected} if \(G[D]\) is connected.  A set \(D\) is a \emph{total dominating set} if every vertex, including every vertex of \(D\), has a neighbor in \(D\):
\[
  N_G(v)\cap D\neq\varnothing
  \qquad\text{for every }v\in V(G).
\]
In particular, every total dominating set is an ordinary dominating set.

For \(\mathcal P\in\{\mathrm I,\mathrm C,\mathrm T\}\), let \(\mathcal D_{\mathcal P}(G)\) denote, respectively, the families of independent, connected, and total dominating sets of \(G\), and define
\begin{align*}
  Z^{\mathcal P}_{\leq}(G,\ell)
  &:=
  \bigl|\{D\in\mathcal D_{\mathcal P}(G):|D|\leq\ell\}\bigr|,\\
  Z^{\mathcal P}_{=}(G,\ell)
  &:=
  \bigl|\{D\in\mathcal D_{\mathcal P}(G):|D|=\ell\}\bigr|.
\end{align*}

For total domination we may assume
\begin{equation}
  d=km\geq2.
  \label{eq:target-at-least-two}
\end{equation}
Indeed, if \(km=1\), duplicate the unique clause.  This preserves the set and number of satisfying assignments and affects neither the asymptotic construction time nor any \(O(k)\) width bound.

\subsection{The canonical sets have the required properties}

\begin{lemma}[Structure of the canonical solutions]
\label{lem:canonical-geometry}
For every matrix \(X\in\F^{k\times k}\), the following hold.
\begin{enumerate}[label=\textup{(\roman*)}]
  \item The set \(D_X\) is independent in \(G_\varphi\).
  \item The graph \(G_\varphi^{\mathrm{sp}}[D_X]\) is the clique \(K_d\).  Consequently, \(D_X\) is connected and, under \eqref{eq:target-at-least-two}, is a total dominating set whenever it dominates \(G_\varphi^{\mathrm{sp}}\).
  \item The graph \(G_\varphi^{\mathrm{bip}}[D_X\cup\{z\}]\) is the star with center \(z\) and leaf set \(D_X\).  Consequently, \(D_X\cup\{z\}\) is connected and is a total dominating set whenever it dominates \(G_\varphi^{\mathrm{bip}}\).
\end{enumerate}
\end{lemma}

\begin{proof}
The set \(D_X\) contains one assignment vertex from each choice group. In \(G_\varphi\), assignment vertices from different choice groups are nonadjacent, and no two vertices of \(D_X\) belong to the same choice clique.  This proves (i).

In the split completion all assignment vertices form one clique, proving the first assertion of (ii).  A clique is connected, and if it has at least two vertices then each of its vertices has a neighbor in the clique. Every vertex outside \(D_X\) is already dominated when \(X\) satisfies the formula, so in that case the total-domination condition holds for all vertices.

In the bipartite construction the assignment vertices are pairwise nonadjacent, while the hub \(z\) is adjacent to every assignment vertex. Thus the indicated induced subgraph is a star.  It is connected; its center has a selected leaf as a neighbor, every selected leaf has the center as a neighbor, and ordinary domination guarantees that every vertex outside the selected star has a selected neighbor.  This proves (iii).
\end{proof}

\subsection{Exact bijections}

\begin{proposition}[Independent domination]
\label{prop:independent-bijection}
The map \(X\mapsto D_X\) is a bijection from the satisfying assignments of \(\varphi\) to the independent dominating sets of \(G_\varphi\) of size at most \(d\).  Every such set has size exactly \(d\), and hence
\begin{equation}
  Z^{\mathrm I}_{\leq}(G_\varphi,d)
  =Z^{\mathrm I}_{=}(G_\varphi,d)
  =\#\operatorname{SAT}(\varphi).
  \label{eq:parsimonious-independent}
\end{equation}
\end{proposition}

\begin{proof}
An independent dominating set is, in particular, a dominating set. Lemma~\ref{lem:basic-bijection} therefore says that every independent dominating set of size at most \(d\) is \(D_X\) for a unique satisfying assignment \(X\).  Conversely, Lemma~\ref{lem:basic-bijection} says that \(D_X\) dominates whenever \(X\) satisfies \(\varphi\), and Lemma~\ref{lem:canonical-geometry}(i) says that it is independent.  The two directions are inverse to one another, proving the proposition.
\end{proof}

\begin{proposition}[Connected and total domination on split graphs]
\label{prop:split-connected-total}
For each \(\mathcal P\in\{\mathrm C,\mathrm T\}\), the map \(X\mapsto D_X\) is a bijection from the satisfying assignments of \(\varphi\) to the members of \(\mathcal D_{\mathcal P}(G_\varphi^{\mathrm{sp}})\) of size at most \(d\). All these sets have size exactly \(d\), and
\begin{equation}
  Z^{\mathcal P}_{\leq}(G_\varphi^{\mathrm{sp}},d)
  =Z^{\mathcal P}_{=}(G_\varphi^{\mathrm{sp}},d)
  =\#\operatorname{SAT}(\varphi).
  \label{eq:parsimonious-split-variants}
\end{equation}
\end{proposition}

\begin{proof}
Every connected or total dominating set is an ordinary dominating set. Proposition~\ref{prop:split-completion} therefore confines every solution of size at most \(d\) to the form \(D_X\) for a unique satisfying assignment \(X\).  Conversely, every satisfying assignment produces a dominating set \(D_X\), and Lemma~\ref{lem:canonical-geometry}(ii) supplies both additional properties.  This proves the claimed bijections and counts.
\end{proof}

\begin{proposition}[Connected and total domination on bipartite graphs]
\label{prop:bip-connected-total}
For each \(\mathcal P\in\{\mathrm C,\mathrm T\}\), the map \(X\mapsto D_X\cup\{z\}\) is a bijection from the satisfying assignments of \(\varphi\) to the members of \(\mathcal D_{\mathcal P}(G_\varphi^{\mathrm{bip}})\) of size at most \(d_{\mathrm{bip}}\).  All these sets have size exactly \(d_{\mathrm{bip}}\), and
\begin{equation}
  Z^{\mathcal P}_{\leq}
  (G_\varphi^{\mathrm{bip}},d_{\mathrm{bip}})
  =Z^{\mathcal P}_{=}
  (G_\varphi^{\mathrm{bip}},d_{\mathrm{bip}})
  =\#\operatorname{SAT}(\varphi).
  \label{eq:parsimonious-bip-variants}
\end{equation}
\end{proposition}

\begin{proof}
Use Proposition~\ref{prop:bip-refinement}(ii) to identify all ordinary dominating sets up to the target, and then use Lemma~\ref{lem:canonical-geometry}(iii) to verify the additional connectedness or total-domination requirement.  Both steps are bidirectional, so no solution is lost or added.
\end{proof}

\subsection{Lower bounds and the precise graph classes}

\begin{proof}[Proof of Theorem~\ref{thm:three-variants}]
For (i), use \(G_\varphi\), the target \(d\), and the order \(\pi\) from \eqref{eq:basic-order}.  Proposition~\ref{prop:independent-bijection} gives correctness, the graph is monopolar by Subsection~\ref{subsec:basic-construction}, and Lemma~\ref{lem:basic-layout} gives the width bound.

For either problem in (ii)--(iii), use \(G_\varphi^{\mathrm{sp}}\) and target \(d\), or use \(G_\varphi^{\mathrm{bip}}\) and target \(d_{\mathrm{bip}}\). Propositions~\ref{prop:split-connected-total} and \ref{prop:bip-connected-total} give correctness.  The graph-class and width statements are exactly Propositions~\ref{prop:split-completion} and \ref{prop:bip-refinement}.

In every case the source has \(k^2\) variables, the output has \(m2^{O(k)}\) vertices, the construction time is \(2^{O(k)}(k+m)^{O(1)}\), and the width is \(O(k)\).  An algorithm with the displayed running time would decide \(\varphi\) in \(2^{o(k^2)}m^{O(1)}\) time, contradicting ETH.  Finally, \(\rw(G)\leq\lrw(G)\leq\width_G(\pi)\) proves the two parameterized formulations, and converting \(\pi\) into a caterpillar rank-decomposition of width \(O(k)\) proves the supplied-decomposition formulation.
\end{proof}

\begin{proof}[Proof of Theorem~\ref{thm:counting-three-variants}]
The bijections underlying \eqref{eq:parsimonious-independent}, \eqref{eq:parsimonious-split-variants}, and \eqref{eq:parsimonious-bip-variants} are parsimonious: each satisfying assignment produces exactly one target-size solution and every solution up to the target arises in this way.  Combining any of these bijections with the size and width calculation from Theorem~\ref{thm:three-variants} would otherwise count the satisfying assignments of \(\varphi\) in \(2^{o(k^2)}m^{O(1)}\) time, contradicting \(\#\mathrm{ETH}\).
\end{proof}

\begin{remark}[Why the restrictions differ]
The split completion cannot be reused for independent domination, because all \(d\) canonical assignment vertices become pairwise adjacent.  In the bipartite construction the forced hub is adjacent to every canonical assignment vertex, so its canonical solution is not independent either. Conversely, precisely these added edges make the canonical solutions connected and total.  Thus the original monopolar construction is tailored to independence, whereas the split and bipartite refinements are tailored to connectivity and total domination.  These are precisely the graph-class conclusions stated in Theorem~\ref{thm:three-variants}.
\end{remark}

\section{The \texorpdfstring{\((\sigma,\rho)\)}{(sigma,rho)} framework}
\label{sec:sigma-rho}

We now extend the quadratic lower bound to two families of \((\sigma,\rho)\)-sets.  A set of nonnegative integers is \emph{cofinite} if its complement in \(\mathbb N_{\ge 0}\) is finite.  In both families, \(\rho\) is cofinite and \(0\notin\rho\).  The latter condition is necessary for a nontrivial minimization problem in the present setting: if \(0\in\rho\), the empty set is always feasible.

There are two ways to ensure that selected vertices satisfy their \(\sigma\)-constraints.  If \(0\in\sigma\) and \(1\in\rho\), we force an independent reservoir and keep the entire canonical solution independent; this places no finiteness assumption on \(\sigma\).  If \(\sigma\) is cofinite, we instead force a large selected clique whose vertices lie in the allowed tail of \(\sigma\).  The second construction produces split graphs.

\begin{theorem}[Cofinite-\(\rho\) \((\sigma,\rho)\)-sets]
\label{thm:sigma-rho-cofinite}
Fix nonempty decidable sets \(\sigma,\rho\subseteq\mathbb N_{\ge 0}\) such that \(\rho\) is cofinite and \(0\notin\rho\).  Suppose that at least one of the following conditions holds:
\begin{enumerate}[label=\textup{(\alph*)}]
  \item \(0\in\sigma\) and \(1\in\rho\);
  \item \(\sigma\) is cofinite.
\end{enumerate}
Unless ETH fails, deciding whether \(G\) has a \((\sigma,\rho)\)-set of size at most \(\ell\) has no
\[
  2^{o(\width_G(\pi)^2)}|V(G)|^{O(1)}
\]
algorithm, even when the order \(\pi\) is supplied.  Under condition \textup{(a)}, the lower bound holds on monopolar graphs.  Under condition \textup{(b)}, it holds on split graphs.  The conclusion also holds when the parameter is linear rank-width, rank-width, or the width of a supplied rank-decomposition.

Assuming \(\#\mathrm{ETH}\), the same lower bounds hold for computing either \(Z^{\sigma,\rho}_{\leq}(G,\ell)\) or \(Z^{\sigma,\rho}_{=}(G,\ell)\).
\end{theorem}
Both branches exploit the cofinite tail of \(\rho\) in the same way. Each clause vertex and equality checker is given
\(q_\rho:=\max(\mathbb N_{\geq0}\setminus\rho)\) neighbors that are forced to be selected. Consequently, its \(\rho\)-constraint fails when none of its assignment neighbors is selected, since it then has exactly \(q_\rho\notin\rho\) selected neighbors, and is satisfied as soon as at least one assignment neighbor is selected, since every count greater than \(q_\rho\) belongs to \(\rho\).
\subsection{The \texorpdfstring{\(0\in\sigma,\ 1\in\rho\)}
{0 in sigma, 1 in rho} branch}
\label{subsec:sigma-rho-zero-one}

Assume throughout this subsection that \(\rho\) is cofinite, \(0\notin\rho\), \(0\in\sigma\), and \(1\in\rho\).  Define
\begin{equation}
  q_\rho:=\max(\mathbb N_{\ge 0}\setminus\rho).
  \label{eq:zero-one-q}
\end{equation}
The maximum exists and is nonnegative because \(\rho\) is cofinite and \(0\notin\rho\).  Hence
\begin{equation}
  q_\rho\notin\rho
  \qquad\text{and}\qquad
  \{q_\rho+1,q_\rho+2,\ldots\}\subseteq\rho.
  \label{eq:zero-one-tail}
\end{equation}

\subsubsection{Independent reservoir and construction}

Choose a constant \(b_0\geq\max\{q_\rho,1\}\).  Create an independent set
\[
  U_0:=\{u_1,\ldots,u_{b_0}\}.
\]
For every \(u\in U_0\), add two vertices \(f_u^0,f_u^1\), each adjacent only to \(u\).  Choose a fixed set \(Q_0\subseteq U_0\) of size \(q_\rho\).

Start with the monopolar graph \(G_\varphi\) of Subsection~\ref{subsec:basic-construction}.  Add the vertices just described, and join every vertex of \(Q_0\) to every clause vertex and every equality checker.  There are no other new edges.  Denote the resulting graph by \(H_\varphi^{0,1}\) and set
\begin{equation}
  d_{0,1}:=b_0+km.
  \label{eq:zero-one-target}
\end{equation}

The graph \(H_\varphi^{0,1}\) is monopolar.  Indeed, the assignment vertices together with the vertices of \(U_0\) induce a cluster graph: the choice groups are its nontrivial clique components and the vertices of \(U_0\) are singleton components.  The guards, clause vertices, checkers, and all \(f_u^i\) form an independent set.

\begin{lemma}[Independent-reservoir normal form]
\label{lem:zero-one-normal-form}
If \(S\) is a \((\sigma,\rho)\)-set of \(H_\varphi^{0,1}\) with \(|S|\leq d_{0,1}\), then
\begin{enumerate}[label=\textup{(\roman*)}]
  \item \(U_0\subseteq S\), and no vertex \(f_u^i\) belongs to \(S\);
  \item exactly one vertex of every \(A_{h,a}\) belongs to \(S\);
  \item no guard, clause vertex, or equality checker belongs to \(S\); and
  \item \(|S|=d_{0,1}\).
\end{enumerate}
\end{lemma}

\begin{proof}
Fix \(u\in U_0\).  If \(u\in S\), the triple \(\{u,f_u^0,f_u^1\}\) contributes at least one selected vertex.  If \(u\notin S\), then an unselected \(f_u^i\) would have no selected neighbor.  Since \(0\notin\rho\), this is impossible, and therefore both \(f_u^0\) and \(f_u^1\) must be selected.  Thus every one of the \(b_0\) disjoint forcing triples contributes at least one vertex, and it contributes exactly one only by selecting \(u\).

Now fix a choice group \((h,a)\).  If \(S\cap A_{h,a}=\varnothing\), then an unselected guard \(g^i_{h,a}\) has no selected neighbor.  Again \(0\notin\rho\), so both guards must be selected.  Otherwise \(A_{h,a}\cup\{g^0_{h,a},g^1_{h,a}\}\) contributes at least one selected vertex.  Hence every one of the \(km\) disjoint choice blocks contributes at least one vertex, and a block with no selected assignment vertex contributes at least two.

The forcing triples and choice blocks are mutually disjoint, so
\[
  |S|\geq b_0+km=d_{0,1}.
\]
The assumed upper bound forces equality.  Every forcing triple must therefore select its center \(u\) and no leaf.  Every choice block must select exactly one assignment vertex; selecting no assignment vertex would require both guards.  These vertices exhaust the budget, excluding all guards, clause vertices, and checkers.  This proves all four claims.
\end{proof}

For a matrix \(X\in\F^{k\times k}\), recall the canonical set \(D_X\) from \eqref{eq:canonical-ds} and put
\begin{equation}
  S_X^0:=U_0\cup D_X.
  \label{eq:zero-one-canonical}
\end{equation}

\begin{lemma}[Bijection for the \(0\in\sigma,\ 1\in\rho\) branch]
\label{lem:zero-one-bijection}
The map \(X\mapsto S_X^0\) is a bijection from the satisfying assignments of \(\varphi\) to the \((\sigma,\rho)\)-sets of \(H_\varphi^{0,1}\) of size at most \(d_{0,1}\).  Every such set has size exactly \(d_{0,1}\), and therefore
\begin{equation}
  Z^{\sigma,\rho}_{\leq}(H_\varphi^{0,1},d_{0,1})
  =
  Z^{\sigma,\rho}_{=}(H_\varphi^{0,1},d_{0,1})
  =\#\operatorname{SAT}(\varphi).
  \label{eq:zero-one-parsimonious}
\end{equation}
\end{lemma}

\begin{proof}
Suppose that \(X\) satisfies \(\varphi\).  No two vertices of \(S_X^0\) are adjacent.  This holds within \(D_X\) by Lemma~\ref{lem:canonical-geometry}(i), within \(U_0\) by construction, and between the two sets because no such edges were added.  Consequently every selected vertex has zero selected neighbors, and its local constraint holds because \(0\in\sigma\).

We next check every unselected vertex.  Each forcing leaf has its center as its unique selected neighbor.  Each guard has the unique selected vertex of its choice group as its neighbor in \(S_X^0\), and every unselected assignment vertex has that same selected vertex as its unique selected neighbor inside the choice clique.  These counts belong to \(\rho\) because \(1\in\rho\).

Finally, a clause vertex or checker has all \(q_\rho\) vertices of \(Q_0\) as selected neighbors.  It also has at least one selected assignment neighbor: for a clause vertex this follows from the satisfaction of its clause, and for a checker it follows from Lemma~\ref{lem:basic-equality}.  Its total selected-neighbor count is therefore larger than \(q_\rho\), and belongs to \(\rho\) by \eqref{eq:zero-one-tail}.  Thus \(S_X^0\) is a \((\sigma,\rho)\)-set of size \(d_{0,1}\).

Conversely, let \(S\) be a \((\sigma,\rho)\)-set of size at most the target.  Lemma~\ref{lem:zero-one-normal-form} gives
\[
  S=U_0\cup
  \{a_{h,a,x_{h,a}}:h\in[m],a\in[k]\}
\]
and excludes all clause vertices and checkers.  Let \(j\) be the number of selected assignment neighbors of one such test vertex.  Its total number of selected neighbors is \(q_\rho+j\).  If \(j=0\), this number is \(q_\rho\notin\rho\), a contradiction.  Hence all clause vertices and checkers have a selected assignment neighbor.  Exactly as in Lemma~\ref{lem:basic-bijection}, the clause vertices show that the matrix of layer \(h\) satisfies \(C_h\), and Lemma~\ref{lem:basic-equality} shows that all layer matrices are equal. Their common value \(X\) satisfies \(\varphi\), and the normal form gives \(S=S_X^0\).  The selected assignment vertices uniquely determine \(X\), so the correspondence is bijective.
\end{proof}

Order the constant-size block
\[
  \Omega_0:=U_0\cup\{f_u^i:u\in U_0,\ i\in\{0,1\}\}
\]
arbitrarily before the order \(\pi\) of \eqref{eq:basic-order}; call the resulting order \(\pi_{0,1}\).

\begin{lemma}[Width of the independent-reservoir construction]
\label{lem:zero-one-width}
The supplied order satisfies
\[
  \width_{H_\varphi^{0,1}}(\pi_{0,1})
  \leq \max\{3b_0,4k+3\}.
\]
\end{lemma}

\begin{proof}
A cut inside \(\Omega_0\) has rank at most \(|\Omega_0|=3b_0\).  At every later cut, all forcing leaves have no crossing edge.  Restricted to the original graph \(G_\varphi\), the crossing matrix has rank at most \(4k+2\) by Lemma~\ref{lem:basic-layout}.  The only remaining crossing edges join the already processed set \(Q_0\) to unprocessed clause vertices or checkers. Their crossing matrix has rank at most one.  Rank subadditivity gives the claimed bound.
\end{proof}

\subsection{The cofinite--cofinite split-graph branch}
\label{subsec:sigma-rho-cofinite-branch}

We now assume that both \(\sigma\) and \(\rho\) are cofinite and that \(0\notin\rho\).  Here the canonical solution will be a large clique, so the tail of \(\sigma\) replaces the assumption \(0\in\sigma\).

\subsubsection{A forced clique reservoir}

Fix cofinite \(\sigma,\rho\subseteq\mathbb N_{\ge 0}\) with \(0\notin\rho\).  Put
\begin{equation}
  q_\sigma:=
  \begin{cases}
    \max(\mathbb N_{\ge 0}\setminus\sigma),
      &\sigma\neq\mathbb N_{\ge 0},\\
    -1,&\sigma=\mathbb N_{\ge 0},
  \end{cases}
  \qquad
  q_\rho:=\max(\mathbb N_{\ge 0}\setminus\rho),
  \qquad
  r:=\min\rho,
  \label{eq:sigma-rho-thresholds}
\end{equation}
The value \(q_\rho\) is well-defined and nonnegative because \(0\notin\rho\).  By definition,
\begin{equation}
  q_\rho\notin\rho,
  \quad
  \{q_\rho+1,q_\rho+2,\ldots\}\subseteq\rho,
  \quad\text{and}\quad
  \{0,\ldots,r-1\}\cap\rho=\varnothing.
  \label{eq:rho-threshold-properties}
\end{equation}

Since \(0\notin\rho\), while \(q_\rho+1\in\rho\) and \(r=\min\rho\), we have
\(1\leq r\leq q_\rho+1\).  Choose a constant
\begin{equation}
	b\geq\max\{q_\rho+1,q_\sigma+1\}.
	\label{eq:reservoir-size}
\end{equation}
The constant depends only on the fixed pair \((\sigma,\rho)\), not on the formula.  Let \(U\) be a clique of size \(b\).  For every \(u\in U\), choose an \(r\)-element set \(R_u\subseteq U\) containing \(u\), and add two pairwise nonadjacent vertices \(f_u^0,f_u^1\) with
\[
  N(f_u^0)=N(f_u^1)=R_u.
\]
All vertices \(f_u^i\) are mutually nonadjacent.  They will force the whole reservoir \(U\) into every solution under the tight budget.
Let
\[
  F:=\{f_u^i:u\in U,\ i\in\{0,1\}\}.
\]

\begin{lemma}[Reservoir lower bound]
		\label{lem:reservoir-lower-bound}
		Let \(S\) be a \((\sigma,\rho)\)-set in any graph containing the
	reservoir described above, and assume that the vertices of \(F\) have no
	neighbors outside \(U\). If \(a:=|U\setminus S|\), then
	\begin{equation}
		|S\cap(U\cup F)|\geq b+a\geq b.
		\label{eq:reservoir-count}
	\end{equation}
	Moreover,
	\[
	|S\cap(U\cup F)|=b
	\]
	if and only if \(U\subseteq S\) and \(S\cap F=\varnothing\).
\end{lemma}

\begin{proof}
	Fix \(u\in U\setminus S\). If \(f_u^i\notin S\), then
	\[
	|N(f_u^i)\cap S|
	=|R_u\cap S|
	\leq r-1,
	\]
	which does not belong to \(\rho\) by
	\eqref{eq:rho-threshold-properties}. Hence both \(f_u^0\) and
	\(f_u^1\) belong to \(S\). Since these forcing vertices are distinct
	for different \(u\), we obtain
	\[
	|S\cap F|\geq 2a.
	\]
	Consequently,
	\[
	|S\cap(U\cup F)|
	=(b-a)+|S\cap F|
	\geq (b-a)+2a
	=b+a.
	\]
	
	Now suppose that \(|S\cap(U\cup F)|=b\). The preceding inequality implies \(a=0\), so \(U\subseteq S\). Therefore
	\[	|S\cap(U\cup F)|=b+|S\cap F|,\]
	and equality with \(b\) implies \(S\cap F=\varnothing\). The converse is immediate.
\end{proof}

\subsubsection{The split-graph construction}

Start with the split graph \(G_\varphi^{\mathrm{sp}}\) from Subsection~\ref{subsec:split-refinement}.  Add the reservoir gadget \(U\cup F\) and make every vertex of \(U\) adjacent to every assignment vertex.  Thus
\begin{equation}
  K_{\sigma,\rho}:=U\cup\bigcup_{h,a}A_{h,a}
  \label{eq:sigma-rho-clique}
\end{equation}
is a clique.

Choose fixed sets \(P,Q\subseteq U\) with
\[
  |P|=r-1
  \qquad\text{and}\qquad
  |Q|=q_\rho.
\]
The choices are possible by \eqref{eq:reservoir-size}; either set may be empty.  Join every guard \(g^i_{h,a}\) to every vertex of \(P\).  Join every clause vertex and every equality checker to every vertex of \(Q\). No other edges are added.  Denote the resulting graph by \(H_\varphi^{\sigma,\rho}\) and set
\begin{equation}
  d_{\sigma,\rho}:=b+km.
  \label{eq:sigma-rho-target}
\end{equation}

All vertices outside the clique \(K_{\sigma,\rho}\) are pairwise nonadjacent: this was true of the guards, clause vertices, and checkers, and the new vertices \(f_u^i\) only have neighbors in \(U\).  Therefore \(H_\varphi^{\sigma,\rho}\) is a split graph.

For an assignment matrix \(X\), define
\begin{equation}
  S_X:=U\cup D_X.
  \label{eq:sigma-rho-canonical-set}
\end{equation}
Here \(D_X\) is the canonical set defined in \eqref{eq:canonical-ds}, viewed in \(G_\varphi^{\mathrm{sp}}\).
\begin{lemma}[Generalized tight-budget normal form]
\label{lem:sigma-rho-normal-form}
If \(S\) is a \((\sigma,\rho)\)-set of \(H_\varphi^{\sigma,\rho}\) with \(|S|\leq d_{\sigma,\rho}\), then:
\begin{enumerate}[label=\textup{(\roman*)}]
  \item \(U\subseteq S\) and no vertex \(f_u^i\) belongs to \(S\);
  \item exactly one vertex of every \(A_{h,a}\) belongs to \(S\);
  \item no guard, clause vertex, or equality checker belongs to \(S\); and
  \item \(|S|=d_{\sigma,\rho}\).
\end{enumerate}
\end{lemma}

\begin{proof}
Lemma~\ref{lem:reservoir-lower-bound} shows that the reservoir and its forcing vertices contribute at least \(b\) selected vertices.

Fix a choice group \((h,a)\).  If \(S\cap A_{h,a}=\varnothing\), then an unselected guard \(g^i_{h,a}\) would have at most \(|P|=r-1\) selected neighbors: its entire neighborhood is \(A_{h,a}\cup P\).  This count is not in \(\rho\), so both guards must be selected.  Consequently,
\begin{equation}
  \bigl|S\cap(A_{h,a}\cup\{g^0_{h,a},g^1_{h,a}\})\bigr|
  \geq1,
  \label{eq:sigma-rho-choice-lower-bound}
\end{equation}
and the contribution is at least two if no assignment vertex is selected. The \(km\) sets in \eqref{eq:sigma-rho-choice-lower-bound} are mutually disjoint and are disjoint from the reservoir part.  Hence
\[
  |S|\geq b+km=d_{\sigma,\rho}.
\]
The \(km\) choice blocks contribute at least \(km\) selected vertices,
so the reservoir block contributes at most \(b\). By
Lemma~\ref{lem:reservoir-lower-bound}, it contributes at least \(b\).
It therefore contributes exactly \(b\), and the final assertion of
Lemma~\ref{lem:reservoir-lower-bound} gives (i). Every choice group must contribute exactly one vertex; the preceding two-guard argument shows that this vertex lies in \(A_{h,a}\), proving (ii) and excluding all guards. The budget is now exhausted, so no clause vertex or checker can be selected.  This proves (iii)--(iv).
\end{proof}

\begin{lemma}[Bijection for cofinite pairs]
\label{lem:sigma-rho-bijection}
The map \(X\mapsto S_X\) is a bijection from the satisfying assignments of \(\varphi\) to the \((\sigma,\rho)\)-sets of \(H_\varphi^{\sigma,\rho}\) of size at most \(d_{\sigma,\rho}\).  Every such set has size exactly \(d_{\sigma,\rho}\).  Consequently,
\begin{equation}
  Z^{\sigma,\rho}_{\leq}
  (H_\varphi^{\sigma,\rho},d_{\sigma,\rho})
  =Z^{\sigma,\rho}_{=}
  (H_\varphi^{\sigma,\rho},d_{\sigma,\rho})
  =\#\operatorname{SAT}(\varphi).
  \label{eq:parsimonious-sigma-rho}
\end{equation}
\end{lemma}

\begin{proof}
Suppose first that \(X\) satisfies \(\varphi\).  The set \(S_X\) has size \(b+km=d_{\sigma,\rho}\), and it lies in the clique \(K_{\sigma,\rho}\). Every selected vertex therefore has exactly
\[
  |S_X|-1=b+km-1\geq b>q_\sigma
\]
selected neighbors.  This number belongs to \(\sigma\).  We verify the \(\rho\)-constraint for each unselected vertex family:
\begin{itemize}
  \item An unselected assignment vertex is adjacent to every vertex of \(S_X\), and hence has \(b+km>q_\rho\) selected neighbors.
  \item A forcing vertex \(f_u^i\) has exactly \(|R_u|=r\) selected neighbors, and \(r\in\rho\).
  \item A guard has its \(r-1\) neighbors in \(P\) and the unique selected vertex of its choice group as neighbors in \(S_X\), for a total of \(r\in\rho\).
  \item A clause vertex or checker has its \(q_\rho\) neighbors in \(Q\) and at least one selected assignment neighbor.  For clause vertices this follows because \(X\) satisfies the clause; for checkers it follows from Lemma~\ref{lem:basic-equality}, since consecutive layers encode the same matrix.  Its selected-neighbor count is therefore strictly larger than \(q_\rho\), and hence belongs to \(\rho\).
\end{itemize}
Thus \(S_X\) is a \((\sigma,\rho)\)-set.

Conversely, let \(S\) be a \((\sigma,\rho)\)-set of size at most the target. Lemma~\ref{lem:sigma-rho-normal-form} gives one matrix \(X_h\) in every layer and excludes every clause vertex and checker from \(S\).  Consider one of these test vertices and let \(j\) be its number of selected assignment neighbors.  Since all of \(Q\) is selected, its total number of selected neighbors is \(q_\rho+j\).  If \(j=0\), this equals \(q_\rho\notin\rho\), contradicting feasibility.  Hence every clause vertex and every checker has at least one selected assignment neighbor. The clause vertices imply that \(X_h\) satisfies \(C_h\), while Lemma~\ref{lem:basic-equality} implies
\[
  X_1=X_2=\cdots=X_m.
\]
The common matrix \(X\) satisfies \(\varphi\), and the normal form gives \(S=U\cup D_X=S_X\).  Uniqueness of \(X\) follows from the unique selected assignment vertex in each choice group, so the correspondence is a bijection.
\end{proof}

\subsubsection{Width of the clique-reservoir construction}

Order the vertices of the reservoir block
\[
  \Omega:=U\cup\{f_u^i:u\in U, i\in\{0,1\}\}
\]
arbitrarily, with all of \(\Omega\) before the order \(\pi\) from \eqref{eq:basic-order}.  Call the resulting order \(\pi_{\sigma,\rho}\).

\begin{lemma}[Width of the clique-reservoir construction]
\label{lem:sigma-rho-width}
The supplied order satisfies
\begin{equation}
  \width_{H_\varphi^{\sigma,\rho}}(\pi_{\sigma,\rho})
  \leq \max\{3b,4k+6\}.
  \label{eq:sigma-rho-width}
\end{equation}
\end{lemma}

\begin{proof}
The block \(\Omega\) has \(3b\) vertices, so a cut inside this block has cut-rank at most \(3b\), independently of the vertices to its right.

Consider a later cut.  Restricted to the original vertices of \(G_\varphi^{\mathrm{sp}}\), its crossing matrix has rank at most \(4k+3\) by Proposition~\ref{prop:split-completion}.  Every forcing vertex \(f_u^i\) is already on the left and has no neighbor outside \(U\), so it creates no crossing edge.  The remaining new crossing edges have one endpoint in \(U\), which is entirely on the left.  A column belonging to an unprocessed original vertex has, on the rows indexed by \(U\), one of only three possible vectors:
\[
  \mathbf 1_U\quad\text{for an assignment vertex},
  \qquad
  \mathbf 1_P\quad\text{for a guard},
  \qquad
  \mathbf 1_Q\quad\text{for a clause vertex or checker}.
\]
Here \(\mathbf 1_S\) denotes the incidence vector of \(S\subseteq U\), with coordinates indexed by \(U\). Their span has dimension at most three.  Rank subadditivity therefore gives cut-rank at most \((4k+3)+3=4k+6\).
\end{proof}

\begin{proof}[Proof of Theorem~\ref{thm:sigma-rho-cofinite}]
Let \(\varphi\) have \(k^2\) variables.  Under condition~\textup{(a)}, construct \(H_\varphi^{0,1}\), use target \(d_{0,1}\), and supply \(\pi_{0,1}\).  Lemma~\ref{lem:zero-one-bijection} gives correctness and parsimony, Lemma~\ref{lem:zero-one-width} gives width \(O(k)\), and the construction is monopolar.  This proves every assertion attached to condition~\textup{(a)}.

Under condition~\textup{(b)}, construct \(H_\varphi^{\sigma,\rho}\), use target \(d_{\sigma,\rho}\), and supply \(\pi_{\sigma,\rho}\).  Lemma~\ref{lem:sigma-rho-bijection} gives correctness and parsimony, Lemma~\ref{lem:sigma-rho-width} gives width \(O(k)\), and the construction is a split graph.  This proves every assertion attached to condition~\textup{(b)}.  If both conditions hold, the two conclusions hold simultaneously on their respective graph classes.

In either construction the pair \((\sigma,\rho)\) is fixed, so the reservoir size is constant.  The output has \(m2^{O(k)}\) vertices and is constructible in \(2^{O(k)}(k+m)^{O(1)}\) time.  A decision or counting algorithm with the prohibited running time would contradict square-variable ETH or \(\#\mathrm{ETH}\), respectively.  Finally, \(\rw(G)\leq\lrw(G)\leq\width_G(\pi)\) for the supplied order of either construction gives the two parameterized formulations; converting the order into a caterpillar rank-decomposition of width \(O(k)\) gives the supplied-decomposition formulation.
\end{proof}

\section{Conclusion and future directions}
\label{sec:conclusion}

We proved that the \(2^{O(\rw(G)^2)}\) dependence of rank-width algorithms for domination is optimal at the level of the exponent under ETH, already for ordinary unweighted Dominating Set.  This resolves the open problem left by Bergougnoux, Korhonen, and Nederlof~\cite{bergougnoux-korhonen-nederlof-2023}.  The obstruction persists on split graphs and on bipartite graphs of diameter at most four.  The same normal form yields a lower bound for Independent Dominating Set on monopolar graphs and for Connected and Total Dominating Set on split graphs and, separately, on bipartite graphs of diameter at most four.  Because the reductions are bijective at the target, all of these statements have direct \(\#\mathrm{ETH}\) counting analogues.  More generally, our \((\sigma,\rho)\)-set theorem applies to every fixed decidable pair with cofinite \(\rho\), \(0\notin\rho\), and either cofinite \(\sigma\) or \(0\in\sigma,1\in\rho\); the latter branch imposes no finiteness assumption on \(\sigma\).  Within the finite/cofinite framework, this covers every cofinite--cofinite pair with \(0\notin\rho\) and every finite--cofinite pair satisfying \(0\in\sigma\) and \(1\in\rho\). Together with the trivial \(0\in\rho\) case, this settles the entire cofinite--cofinite minimization regime.

\subsection{Perfect Dominating Set and Perfect Code}

Two basic exact-neighbor problems remain outside our lower bounds.  Perfect Dominating Set and Perfect Code correspond, respectively, to
\[
(\sigma,\rho)=(\mathbb N_{\ge0},\{1\})
\qquad\text{and}\qquad
(\sigma,\rho)=(\{0\},\{1\}).
\]
In both problems, every unselected vertex must have \emph{exactly one} selected neighbor; Perfect Code additionally requires the selected vertices to form an independent set.  The reductions in this paper deliberately test only whether a clause vertex or equality checker has at least one selected assignment neighbor.  Cofinite padding is ideal for that monotone test: it separates zero from every positive count.  It cannot enforce that all positive counts collapse to the single allowed value \(1\).

A lower bound for Perfect Dominating Set therefore needs an \emph{exact-one checker}.  The challenge is to prevent multiple selected assignment neighbors without exposing \(\Theta(k^2)\) independent information across a cut and thereby losing the required \(O(k)\) rank bound.  To the best of our knowledge, no matching rank-width lower bound is known for Perfect Dominating Set.  The general upper bound follows from neighborhood-union dynamic programming~\cite{bui-xuan-telle-vatshelle-2013,oum-saether-vatshelle-2014}; in particular, the lower bounds of Bergougnoux, Korhonen, and Nederlof do not cover this problem~\cite{bergougnoux-korhonen-nederlof-2023}.

Perfect Code is likewise an open case for the present framework.  Split graphs, however, cannot provide hard instances: a split partition can be found in linear time~\cite{hammer-simeone-1981}, and Appendix~\ref{app:perfect-code} gives a complete characterization that permits implicit representation and counting of all perfect codes, as well as finding minimum- and maximum-cardinality perfect codes, in linear time.  Any extension of our lower-bound approach to Perfect Code must therefore use a different host graph class.

\subsection{The three remaining finite/cofinite regimes when
\texorpdfstring{\(0\notin\rho\)}{0 not in rho}}

Using neighborhood-union dynamic programming, the minimum-cardinality \((\sigma,\rho)\)-set problem for every fixed pair in which each of \(\sigma\) and \(\rho\) is finite or cofinite can be solved in \(2^{O(k^2)}|V(G)|^{O(1)}\) time from a supplied rank-decomposition of width \(k\)~\cite{bui-xuan-telle-vatshelle-2013,oum-saether-vatshelle-2014}.  For the minimization decision problem, the case \(0\in\rho\) is trivial because the empty set is feasible.  Thus, in the natural nontrivial setting \(0\notin\rho\), our matching lower bound completely settles the cofinite--cofinite minimization regime.  Exact-size and counting problems with \(0\in\rho\) form a separate classification question.  Under \(0\notin\rho\), the three remaining finite/cofinite regimes are not yet fully covered; their structural obstacles are as follows.

\begin{description}[style=nextline]
	\item[Finite--finite.] Both selected and unselected vertices must have bounded allowed numbers of neighbors in the solution.  Neither large-clique validation nor cofinite padding is available, so both sides of the local constraint must be controlled simultaneously while the cut-rank remains \(O(k)\).
	
	\item[Cofinite--finite.] The cofinite tail of \(\sigma\) can validate selected vertices, but the finite set \(\rho\) requires every unselected test vertex to have one of finitely many exact neighbor counts.  Thus the unresolved difficulty lies on the unselected side of the constraint.
	
	\item[Finite--cofinite.] The cofinite tail of \(\rho\) still permits the padding test used here, but selected vertices cannot in general be placed in a large clique.  Section~\ref{subsec:sigma-rho-zero-one} solves the subfamily \(0\in\sigma,1\in\rho\) by keeping the canonical solution independent.  For the remaining pairs, the selected side must be controlled without either a large clique or the assumptions \(0\in\sigma\) and \(1\in\rho\).
\end{description}

\clearpage
\appendix

\section{A standard-basis equality checker}
\label{app:standard-checker}

The checker in Section~\ref{subsec:basic-construction} can be made smaller. Let \(\mathbf e_j\) denote the \(j\)-th standard basis vector of \(\F^k\). For every transition \(h\in[m-1]\), replace \(E_h\) by the standard-basis checker block
\[
  E_h^{\mathrm{sb}}
  :=\{e_{h,j,p,r}:j\in[k],\ p\in\F^k,\
  r\in\F^k\setminus\{0\}\}.
\]
Join \(e_{h,j,p,r}\) to \(a_{h,a,x}\) when \(x_j\neq p_a\), and to \(a_{h+1,a,x}\) when \(x_j\neq p_a+r_a\).

\begin{lemma}
\label{lem:standard-checker}
If consecutive layers encode \(X,Y\in\F^{k\times k}\), then every standard-basis checker between them has a selected assignment neighbor if and only if \(X=Y\).
\end{lemma}

\begin{proof}
The checker \(e_{h,j,p,r}\) has no selected assignment neighbor precisely when
\[
  X\mathbf e_j=p\quad\text{and}\quad Y\mathbf e_j=p+r.
\]
If \(X=Y\), these equations force \(r=0\).  If \(X\neq Y\), some column \(j\) differs; setting \(p=X\mathbf e_j\) and \(r=(X+Y)\mathbf e_j\neq0\) produces a checker with no selected assignment neighbor.
\end{proof}

There are only \(k2^k(2^k-1)\) checkers per transition.  The left incidence matrix has entries \(x_j+p_a\); its rows lie in the span of the \(k\) functions \(\mathbf 1[j=\ell]\) and the \(k\) functions \(p_a\). Its rank is therefore at most \(2k\).  The right incidence matrix, with entries \(x_j+p_a+r_a\), has the same bound.  Replacing \(E_h\) by \(E_h^{\mathrm{sb}}\) in the proofs of Section~\ref{sec:unweighted-ds} therefore preserves the equality test, all target-size bijections, and the order-width bounds.  The graph-class arguments are unchanged; for the bipartite diameter bound, note that every standard-basis checker has an assignment neighbor in each endpoint layer.  The number of vertices drops to \(O(mk2^{2k})\), and every refinement can still be generated explicitly in \(2^{O(k)}(k+m)^{O(1)}\) time.  We use this checker below.

\section{A \texorpdfstring{\(3k+O(1)\)}{3k+O(1)} rank-decomposition}
\label{app:rank-decomposition}

\begin{theorem}
\label{thm:three-k-decomposition}
For a source formula on \(k^2\) variables, the basic graph constructed using the standard-basis checker from Appendix~\ref{app:standard-checker} has a polynomial-time constructible rank-decomposition of width at most \(3k+1\), and at most \(3k\) when \(k\geq2\).  The split and bipartite refinements have polynomial-time constructible rank-decompositions of width at most \(3k+2\).
\end{theorem}

\begin{proof}
Use the block order
\[
  \Lambda_1,E_1^{\mathrm{sb}},\Lambda_2,E_2^{\mathrm{sb}},\ldots,
  E_{m-1}^{\mathrm{sb}},\Lambda_m.
\]
Make a caterpillar for these blocks and replace each block leaf by a rooted caterpillar containing its vertices.  Inside \(\Lambda_h\), use the order from Section~\ref{sec:unweighted-ds}; inside \(E_h^{\mathrm{sb}}\), use any order.  The roots are chosen so that a non-leaf cut inside a layer detaches a prefix of that order.  This gives a subcubic tree with one leaf per graph vertex.  The attachment edge of a rooted block detaches the entire block and is included in the corresponding inside-block case below.

Three kinds of cuts occur.  First, a backbone cut lies between two whole blocks.  Only one layer--checker interface crosses it, so its rank is at most \(2k\).

Second, a cut inside \(E_h^{\mathrm{sb}}\) detaches a subset of that checker block.  Over the two endpoint layers, every checker row is the sum of one of \(k\) column-selector vectors, a linear combination of \(k\) row-selector vectors present in both layers, and a linear combination of \(k\) row-selector vectors supported only on the later layer.  Hence the combined incidence matrix has rank at most \(3k\).

Third, consider a cut inside \(\Lambda_h\).  With coordinates indexed by \(E_{h-1}^{\mathrm{sb}}\cup E_h^{\mathrm{sb}}\), let \(f_\ell\) be the binary vector that records whether a checker has basis index \(j=\ell\), in either block.  Let \(q_a\) be the binary vector whose coordinate is \(p_a+r_a\) on \(E_{h-1}^{\mathrm{sb}}\) and \(p_a\) on \(E_h^{\mathrm{sb}}\).  The row of an assignment vertex \(a_{h,a,x}\) toward these two checker blocks is
\[
  \sum_{\ell:x_\ell=1} f_\ell+q_a,
\]
so both checker interfaces together have rank at most \(2k\).  A prefix cut inside the layer increases the rank by at most one for the split choice group and by at most one for the clause vertex, exactly as in Lemma~\ref{lem:basic-layout}; a leaf cut is no larger.  Its total rank is at most \(2k+2\).

The width is therefore
\[
  \max\{2k,3k,2k+2\}\leq3k+1,
\]
and is at most \(3k\) for \(k\geq2\).

For the split refinement, use the same tree.  A backbone cut has rank at most \(2k+1\): one checker interface contributes at most \(2k\), and the global assignment clique increases the rank by at most one.  A cut inside a checker block still has rank at most \(3k\).  A cut inside a layer has rank at most \(2k+3\): the two checker interfaces together contribute at most \(2k\), while the global assignment clique, the split guard--assignment incidence, and the clause vertex each increase the rank by at most one.

For the bipartite refinement, attach at one end of the backbone a rooted three-leaf block containing \(z,z^0,z^1\).  Cuts inside this block have rank at most one.  For every other cut, the hub star increases the rank by at most one; on a cut that detaches part of a checker block, the hub and all assignment vertices lie on the same side, so the hub contributes no crossing edge.  Thus backbone cuts have rank at most \(2k+1\), checker-block cuts have rank at most \(3k\), and layer-block cuts have rank at most \(2k+3\): besides the two checker interfaces, only the hub star, the split guard--assignment incidence, and the clause vertex need be counted.  Consequently both refinements have width at most
\[
  \max\{2k+1,3k,2k+3\}\leq3k+2.
\]
\end{proof}

\section{Quantitative \texorpdfstring{\(\#\mathrm{SETH}\)}{\#SETH} constants}
\label{app:seth-constants}

We record constants only for counting unweighted Dominating Set.  We use the standard counting SETH formulation: for every \(\delta\in(0,1)\), there exists a clause width \(q\) for which no \(2^{(1-\delta)N}N^{O(1)}\)-time algorithm counts the satisfying assignments of \(q\)-CNF formulas on \(N\) variables~\cite{curticapean-marx-2016,focke-et-al-2023}.

\begin{theorem}
\label{thm:quantitative-counting}
Assume \(\#\mathrm{SETH}\).  The following statements hold on the basic monopolar instances.
\begin{enumerate}[label=\textup{(\roman*)}]
\item For every \(\varepsilon\in(0,1/9)\), when a rank-decomposition of width \(w\) is supplied, neither \(\Zle(G,d)\) nor \(\Zeq(G,d)\) can be computed in time
\[
  2^{(1/9-\varepsilon)w^2}|V(G)|^{O(1)}.
\]
\item For every \(\varepsilon\in(0,1/16)\), when a vertex order of width \(w\) is supplied, neither quantity can be computed in time
\[
  2^{(1/16-\varepsilon)w^2}|V(G)|^{O(1)}.
\]
\end{enumerate}
\end{theorem}

\begin{proof}
For the first claim, set \(\eta=3\varepsilon\) and let \(q\) be the clause width guaranteed by \(\#\mathrm{SETH}\) for the saving \(\delta=3\varepsilon\).  Disjoint counting sparsification produces at most \(2^{\eta N}\) sparse \(q\)-CNF formulas on a common \(N\)-variable universe whose satisfying-assignment sets are pairwise disjoint~\cite[Appendix~A]{dell-et-al-2014}.  Formulas with an empty clause or with no clauses are counted directly; every remaining formula has \(m=O(N)\) nonempty clauses.  The construction does not use the bound three on clause size, so apply it to each remaining sparse formula and sum the resulting counts.  Pad each formula parsimoniously to \(k^2\) variables, where \(k=\lceil\sqrt N\rceil\).  By Lemma~\ref{lem:standard-checker} and the proof of Lemma~\ref{lem:basic-bijection}, the standard-basis construction is parsimonious; it has \(m2^{O(k)}=2^{o(N)}\) vertices.  Appendix~\ref{app:rank-decomposition} supplies width \(w\leq3k+1\).  On each sparse instance, the first hypothetical algorithm would therefore count satisfying assignments in time
\[
  2^{(1/9-\varepsilon)(3k+1)^2+o(N)}
  =2^{(1-9\varepsilon)N+o(N)}.
\]
The total running time over all sparsified instances is \(2^{(1-6\varepsilon)N+o(N)}\), which for sufficiently large \(N\) is at most \(2^{(1-3\varepsilon)N}\), contradicting the choice of \(q\).  For the second claim, repeat the argument with \(\eta=4\varepsilon\) and the clause width guaranteed by \(\#\mathrm{SETH}\) for \(\delta=4\varepsilon\).  The standard-basis analogue of the order in Lemma~\ref{lem:basic-layout} has width at most \(4k+2\), giving a total running time of \(2^{(1-12\varepsilon)N+o(N)}\), which is eventually at most \(2^{(1-4\varepsilon)N}\), again a contradiction.
\end{proof}

\section{Perfect Code on split graphs}
\label{app:perfect-code}

A perfect code is a \((\{0\},\{1\})\)-set.  Eschen and Wang observed the same two-case structure for perfect codes on split graphs~\cite{eschen-wang-2014}.  We state the characterization explicitly because it additionally yields a linear-time implicit representation and count, as well as minimum- and maximum-cardinality solutions.

\begin{theorem}
\label{thm:perfect-code-split}
Given a split graph with partition \(V(G)=C\mathbin{\dot\cup}I\), all its perfect codes can be represented implicitly and counted in \(O(|V(G)|+|E(G)|)\) time.  Within the same time bound, a minimum- and a maximum-cardinality perfect code can be found whenever one exists.
\end{theorem}

\begin{proof}
Let
\[
  I_0:=\{u\in I:N(u)=\varnothing\},
  \qquad I_+:=I\setminus I_0.
\]
The complete list of perfect codes is
\begin{equation}
  \begin{cases}
    \{I\},&|N(c)\cap I|=1\text{ for every }c\in C,\\
    \varnothing,&\text{otherwise}
  \end{cases}
  \;\cup\;
  \bigl\{\{c\}\cup I_0:c\in C,\ I_+\subseteq N(c)\bigr\}.
  \label{eq:split-perfect-codes}
\end{equation}

Indeed, a perfect code \(D\) is independent, so \(|D\cap C|\leq1\).  If \(D\cap C=\varnothing\), every vertex of \(I\) must belong to \(D\); hence \(D=I\), and every clique vertex must have exactly one neighbor in \(I\). If \(D\cap C=\{c\}\), every nonisolated vertex of \(I\) must be adjacent to \(c\): otherwise it must belong to \(D\), which would give one of its clique neighbors both \(c\) and that vertex in its closed neighborhood. Conversely, when \(I_+\subseteq N(c)\), the set \(\{c\}\cup I_0\) meets every closed neighborhood exactly once.  This proves \eqref{eq:split-perfect-codes}.

One adjacency-list scan finds \(I_0\), tests whether every \(c\in C\) has exactly one neighbor in \(I\), and computes \(|N(c)\cap I_+|\) for every \(c\).  It therefore enumerates the displayed candidates implicitly, counts them, and compares their sizes \(|I|\) and \(1+|I_0|\) to find minimum- and maximum-cardinality candidates or to test any given target.  The running time is linear.
\end{proof}

\end{document}